\documentclass[12pt,oneside,a4paper]{book}
\usepackage[utf8]{inputenc}
\usepackage{amsmath}
\usepackage{cite}
\usepackage{enumerate}
\usepackage{hyperref}
\usepackage{fancyhdr}
\usepackage{graphicx}
\usepackage{array}
\usepackage{amsmath,amsfonts,amssymb,amsthm,epsfig,epstopdf,url,array}
\usepackage[toc,page]{appendix}
\usepackage{setspace} 
\usepackage[autostyle]{csquotes}
\MakeAutoQuote{‘}{’}
\theoremstyle{plain}
\newtheorem{thm}{Theorem}[section]
\newtheorem{lem}[thm]{Lemma}
\newtheorem{prop}[thm]{Proposition}

\newcommand*\closure[1]{\overline{#1}}

\theoremstyle{definition}
\newtheorem{defn}{Definition}[section]

\theoremstyle{remark}

\theoremstyle{plain}
\newtheorem{theorem}{Theorem}[section]

\newtheorem{lemma}[theorem]{Lemma}

\usepackage{amssymb}
\usepackage{graphicx}
\usepackage[toc,section=section]{glossaries}

\newglossary{abbrev}{abs}{abo}{List of Abbreviations}
\newglossaryentry{MS}{
	name        = MS ,
	description = mass spectroscopy ,
	type        = abbrev
}

\usepackage{array}
\usepackage{rotating}
\renewcommand{\geq}{\geqslant}
\renewcommand{\leq}{\leqslant}

\usepackage{amsthm,amsmath,amssymb,latexsym}

\theoremstyle{plain}
\theoremstyle{definition}

\numberwithin{equation}{section}

\usepackage{amsmath}               
\theoremstyle{assumption}
\newtheorem{assumption}{Assumption}
\begin{document}
\singlespacing
\thispagestyle{empty}
\begin{center}
      {\Large Estimation of the number of negative eigenvalues of magnetic Schr\"{o}dinger operators in a strip}\\
      \vspace{1.8in}
      {\large Sorowen Ben}\\
      \vspace{0.15in}
      {\large 2018/MSc/019/PS}\\
      \vspace{0.15in}
      {\large BSc. Edu, Kyambogo University, 2014.}\\
      \vspace{1.5in}
      {\Large A dissertation submitted in partial fulfilment of the requirements leading to the award of the degree of Master of Science in Mathematics of Mbarara University of Science and Technology }
\end{center}
\vspace{0.4in}
\begin{center} 
    February 11, 2021  
\end{center}
 \frontmatter
 \newpage
 \clearpage

\section*{Declaration}
     I hereby declare that the work  presented in this dissertation is my own and has not been presented to this institution or any other institution for any award. I confirm that where I have used materials from other sources, references have been made and with the exception of such, this dissertation is my original work.\\
     
     Name:\rule{14cm}{0.002in}\\

     Signature:\rule{7cm}{0.002in}
     Date:\rule{5cm}{0.002in}\\

\newpage
\section*{Approval}
     The research work culminating in this dissertation was conducted under my  guidance and supervision.

     Dr. Martin Karuhanga\\
     Department of Mathematics\\
     Faculty of Science\\
     Mbarara University of Science and Technology\\

     Signature:\rule{7cm}{0.002in}
     Date:\rule{5cm}{0.002in}\\

\newpage
\section*{Dedication}
To my Parents  Mr. Ben Saik and Mrs. Kapcherop Janerose  Saik.

\newpage
\section*{Acknowledgment}
       I thank the Almighty God for enabling me complete my course  successful and this dissertation. 
       
       Great thanks go to my supervisor Dr. Martin Karuhanga for the useful discussions and valuable time, mentorship and inspirations rendered to me from time to time. I am grateful to Dr. Feresiano for the guidance he gave me while compiling this dissertation and I extend my sincere thanks to the staff members of department of mathematics of Mbarara University of Science and Technology and that of Kyambogo University for their guidance support.
       
       I am greatly Indebted to the Pure mathematicians Dr. John Emenyu, Mr. Kimuli Philly Ivan (PhD candidate at Makerere University), Mr. Wilbroad Bezire (PhD student at Moi University) and Mr. Sentayi Mbidde Abdulwahabu, the physicists Mr. Habumugisha Isaac, Mr. Ndugu Nelson (both PhD candidates at Mbarara University of Science and Technology) and Mr. Andama Geoffrey (PhD student at Mbarara University of Science and Technology), Dr. Joseph Ssenyonga, the laboratory technicians of Mbarara University of Science and Technology Mr. Cherop Tonny, Mr. Johnathan Baguma and Mr. Araka Nickson, the head of department of mathematics at Mbarara University Prof. Julius Tumwiine, and that of Kyambogo University Dr. Hasifa Nampala and the former head of department of mathematics Kyambogo University Mr. Tengi Jacob for their support and guidance.
       
       I owe sincere thanks to the administration of Kyambogo University, that of Mbarara university of Science and Technology and that of the Regional Universities Forum for Capacity Building in Agriculture (RUFORUM) for the study scholarship under the Graduate Teacher Assistant programme.
       
       Finally, I appreciate my brothers, sisters and parents for their love, support and perseverance exercised during my study.

\newpage
\section*{Abstract}
     An upper estimate for the number of negative eigenvalues below the essential spectrum for the magnetic Schr\"{o}dinger operator with Aharonov-Bohm magnetic field in a strip is obtained.

     Its further shown that the estimate does not hold in absence of Aharonov-Bohm magnetic field. 

\newpage
\section*{Basic Notation}
\begin{itemize}
	\item $\mathbb{R}$ denotes a field of real numbers.
	\item $\mathbb{R}^{n}$ denotes the Euclidean $n$-space.
	\item $\Omega$ denotes an arbitrary open set.
	\item $\mathcal{H}$ denotes a complex Hilbert space.
	\item $W^{m,p}(\Omega)$ denotes a Sobolev space and when $p=2$ it takes the standard notation $H^{2}(\Omega)$.
	\item $L^{p}(\Omega)$ denotes a function space on $\Omega$.
	\item $l^{p}(\Omega)$ denotes a sequence space on $\Omega$.
	\item $C^{1}(\Omega)$ denotes the space of continuously differentiable functions on $\Omega$.
	\item $C^{n}(\Omega)$ denotes a unitary space on $\Omega$.
	\item $\closure{\Omega}$ denotes closure of $\Omega$.
	\item $C_{0}^{\infty}(\Omega)$ denotes infinitely differentiable functions with compact support in $\Omega$.
	\item $L^{2}(\mathbb{R}^{n})$ denotes the space of square integrable functions on $\mathbb{R}^{n}$.
	\item $\partial \Omega$ denotes the boundary of an open set $\Omega$.
	\item $\nabla$ defines the gradient.
	\item $T^{*}$ denotes the adjoint of a linear differential operator $T$.
	\item $\langle.,.\rangle$ denotes the inner product.
	\item $\|.\|$ denotes the norm of a bounded linear operator.
	\item $L_{loc}^{1}(\Omega)$ denotes  the space of locally integrable functions  in $\Omega$.
	\item $Neg(H)$ denotes the number of negative eigenvalues for the operator $H$.
	\item $\mathcal{D}(T)$ denotes the domain of an operator $T$.
	\item $I$ denotes the identity operator.
	\item $a\wedge b$ denotes the exterior product of $a$ and $b$.
\end{itemize}

\newpage
\tableofcontents
\mainmatter
\begin{chapter}{Introduction}
\begin{section}{Background}
	  The study of plasma, nuclear, solid state, semiconductor physics, and astrophysics \cite{bonitz2019quantum,goswami2020energy,janev2016review,rezaei2016saturation,weisheit1989atomic} requires the knowledge about the number of bound-states solutions of the Schr\"{o}dinger equation.
      The differential operator appearing in the Schr\"{o}dinger equation is called Schr\"{o}dinger operator and has got deep roots in non-relativistic quantum mechanics \cite{capri2002nonrelativistic}. The Schr\"{o}dinger operator in a specific quantum system offers a  full characterization of its physical system. Eigenvalues and eigenfunctions of an operator usually have a direct physical interpretation of the physical system it models, for example, eigenvalue estimates for magnetic Schr\"{o}dinger operators play an important role in the study of asymptotic properties of the discrete spectrum \cite{last2005spectral}, scattering theory \cite{taylor2006scattering}, and stability of matter in magnetic fields \cite{lieb1997stability}. This therefore makes the study of spectral properties of Schr\"{o}dinger operators an interesting object in Mathematical Physics.
      
      A number of mathematicians inspired by the need to come up with generalizations of the Schr\"{o}dinger operator in the 1950s, posed a number of questions concerning the Schr\"{o}dinger operator $H_{V}=-\Delta+V$, $V\geq 0$ defined on a Hilbert space $L^{2}(\mathbb{R}^{n})$. Such questions include defining the Schr\"{o}dinger operator as a self adjoint operator in a proper Hilbert space and determining its spectral properties on quantitative and qualitative levels.
      
      Complete set of answers to these questions were found with regularity conditions on the potential and more relaxed with time \cite{reed1978analysis}. Some result such as bounds on the number of negative eigenvalues of the operator were obtained by  Rozenblum \cite{rozenblum1972distribution}, Lieb \cite{lieb1976bounds} and  Cwikel \cite{cwikel1977weak}. Today, this result is known as the Cwikel-Lieb-Rozenblum  inequality (CLR) \cite{balinsky2011spectral}.
      
      Later, the researchers directed their interest also to the operators involving both electric field  potential and the magnetic fields such as the magnetic  Schr\"{o}dinger operators and the Pauli operators. Earlier studies found in the work of Avron et al. \cite{avron1978schrodinger}, Reed\cite{reed1978analysis}, Simon\cite{simon1979functional},  Leinfelder and Simader\cite{leinfelder1981schrodinger} \& Melgaard and Rozenblum \cite{melgaard1996spectral} lead to obtaining the analogous of the CLR inequality for the magnetic inequality. Spectral theory problems of the  magnetic Schr\"{o}dinger operator to date have not been fully characterized  such as regularity of the potential discussed in view of Barseghyan and Truc \cite{barseghyan2017magnetic}.
      Problems such as obtaining the bounds for the number of negative eigenvalues for the Schr\"{o}dinger operator in two dimension  to this moment is  still an open problem (see \cite{grigor2014lower,karuhanga2018eigenvalue,karuhanga2019discrete,karuhanga2020negative,shargorodsky2014negative} and the references therein).
      
      The spectral theory of Schr\"{o}dinger operators has during the past been restricted  only to the qualitative questions of spectral theory, self adjoint realization, semi-boundedness, discreteness of the whole of the negative spectrum while the quantitative results on the spectrum lag behind to this moment.      
      
      In this study, the question of obtaining an estimate for the number of negative eigenvalues for the magnetic Schr\"{o}dinger operator in a strip subject to Neumann boundary conditions is considered. A related problem was solved in \cite{balinsky2001number} in the case of the Aharonov-Bohm magnetic field represented by a Dirac delta function and in \cite{kovarik2011eigenvalue} for a reasonably large class of magnetic fields. However, in the two cases considered, they do not consider the case of a strip where the boundary conditions affect the bottom of the essential spectrum and estimates obtained in \cite{balinsky2001number} fail in regular fields.
      
      One of the tools used to obtain the bound for the number of eigenvalues is the diamagnetic inequality. In \cite{rozenblum2001semigroup}, its pointed out that using only the  diamagnetic inequality alone to derive eigenvalue estimates for the magnetic Schr\"{o}dinger operator from the non-magnetic ones would, in particular, lead automatically to estimates not depending on the magnetic field. Thus obtaining eigenvalue estimates for the magnetic Schr\"{o}dinger operator becomes a separate problem, and many methods used in the non-magnetic situation do not work in the magnetic case. In this study the Hardy inequality, Bargmann type inequality and the Sturm-Liouville operator theory are particularly used to obtain the eigenvalue estimates for the magnetic Schr\"{o}dinger operator in a strip with Aharonov-Bohm type magnetic field.
\end{section}
\begin{section}{Structure of the thesis}
	In Chapter 1, we give a brief background to the problem and define the magnetic Schr\"{o}dinger operator, we give the statement of the problem, state the aims, objectives and justification of the study. We also review the basic facts in spectral theory and introduce the quadratic form of a self-adjoint operator.
	
	In Chapter 2, we review some of the preliminaries, auxilliary results and present the literature review. 
	
	In Chapter 3, we obtain an estimate for the number of negative eigenvalues of a magnetic Schr\"{o}dinger operator in a strip with Aharonov-Bohm magnetic field. The estimate only holds in the presence of magnetic field of the Aharonov-Bohm type with Neumann conditions subjected at the boundary of the strip, and the magnetic flux satisfying the non-integer condition.
	
	In Chapter 4, we give some recommendations and suggestions for further study.
\end{section}
\begin{section}{Magnetic Schr\"{o}dinger operators}
\begin{subsection}{Magnetic Laplacian}
    The classical Laplace operator $H_{0}=-\Delta$  defined by
    \begin{equation*}
    H_{0} =-\Delta=-\sum_{i=1}^{n}\frac{\partial^{2}}{\partial x_{i}^{2}} \text{ on }  L^{2}(\mathbb{R}^{n})
    \end{equation*} 
    is a positive self-adjoint operator whose  spectrum is real and its essential spectrum is absolutely continuous, that is,  $\sigma(-\Delta)=[0,\infty)$. The discrete spectrum is empty \cite{last2005spectral}, describes the energy of a particle in an electric field.
    
    To describe the energy of a particle in an external magnetic field $B:\mathbb{R}^{n}\to \mathbb{R}^{n}$, the Laplacian $\Delta=\nabla^{2}$ is replaced by $(\frac{1}{i}\nabla+A)^{2}$, where the vector potential $A:\mathbb{R}^{n}\to \mathbb{R}^{n}$ satisfies $curl A=B$ \cite{hansson2007spectral}. We denote  the magnetic Laplacian by $H _{B}$. Then
    \begin{equation}\label{eq:ben3}
    H_{B}=\bigg( \frac{1}{i}\nabla +A \bigg)^{2} \text{\;on\;} L^{2}(\mathbb{R}^{n}).
    \end{equation}
    It is shown in \cite{borisov2005spectrum} that the operator $H_{B}$ is non-negative  and symmetric in $L^{2}(\mathbb{R}^{n})$  and therefore it can be extended to self adjoint non-negative operator. Balinsky et al \cite{balinsky2001number} proved that the essential spectrum of the operator $H_{B}$ is $[0, \infty)$  (see \cite[Theorem 6.1]{cycon2009schrodinger}) and its discrete spectrum  is empty if the potential $A$ is of the Aharonov-Bohm type (see \cite[Section 5.5]{balinsky2015analysis} or \cite[Section 6]{cycon2009schrodinger},  for more details).
\end{subsection}
\begin{subsection}{Magnetic Schr\"{o}dinger operator}
    Magnetic Schr\"{o}dinger operators  appear in several areas of Physics, such as the Ginzburg-Landau theory of superconductors \cite[Section 10]{fournais2010spectral}, \cite{berdiyorov2012magnetoresistance,elmurodov2008phase}, the theory of Bose-Einstein condensates \cite{jun2004exact,kapitula2005bose}, and the study of edge states in quantum mechanics \cite{raghu2008analogs,frank2007asymptotic}.
    The magnetic Schr\"{o}dinger operator  is  defined by 
    \begin{equation}\label{eq:ben4}
    H_{A,V}=\bigg( \frac{1}{i}\nabla+A\bigg)^{2}-V \text{\;on\;} L^{2}(\mathbb{R}^{n}) ,
    \end{equation}
    where $A$ is a magnetic potential and $0\leq V\in{L_{loc}^{1}(\mathbb{R}^{n})}$ is the electric potential \cite{balinsky2011spectral,chung2014partial}.\\
    Under certain assumptions on $A$ and $V$, the operator $H_{B,V}$ is essentially self-adjoint and has a discrete spectrum \cite{ekholm2008stability}.
 \end{subsection}
 \end{section}
\newpage
\begin{section}{Statement of the problem}
    Given a function  $V:\mathbb{R}^{n}\to{[0,\infty)}$, consider the Schr\"{o}dinger operator
    \begin{equation}\label{eq:ben2}
     H_{V}:=-\Delta-V \text{\;on\;} L_{2}(\mathbb{R}^{n}), 
    \end{equation} 
    under certain assumptions, about $V$, such as if $V $ is uniformly locally square integrable for $n\leq 3$ or is uniformly locally in $L^{p}$ with $p>\frac{n}{2}$ for $n>3$, the operator \eqref{eq:ben2} is well-defined and self-adjoint on $L^{2}(\mathbb{R}^{n})$. Moreover, if $V(x)\to{0}$ as $|x|\to{\infty}$, its essential spectrum is the interval $[0,\infty)$ and 
    the discrete spectrum consists of eigenvalues of finite multiplicity \cite{last2005spectral}.
    In quantum mechanics, we call these negative eigenvalues the bound states and they describe energy levels of a quantum system, their respective eigenfunctions correspond to energy states \cite{simon1982schrodinger}. According to the CLR inequality, the number of negative eigenvalues of $H_{V}$ for  $n\geq 3$ is estimated by
    \begin{equation}\label{eq:ben6}
    Neg(V,\mathbb{R}^{n})\leq C\int_{\mathbb{R}^{n}}V(x)^{n/2}dx ,
    \end{equation}
    where $Neg(V,\mathbb{R}^{n})$ denotes the number of negative eigenvalues and $C$ a constant which depends on the dimension. It is known that the CLR inequality does not hold for $n=2$ (see \cite{shargorodsky2014negative,karuhanga2018eigenvalue,karuhanga2019discrete,karuhanga2020negative} and the reference there in). Some of the reasons are that the Sobolev space $H^{1}(\mathbb{R}^{2})$ is not continuously embedded in $L^{\infty}(\mathbb{R}^2)$ \cite{karuhanga2017estimates,karuhanga2018eigenvalue}, the Hardy inequality  
    \begin{equation}\label{eq:ben7}
    \int_{\mathbb{R}^{n}}\frac{|u(x)|^{2}}{|x|^{2}}dx\leq \bigg(\frac{n-2}{2}\bigg)^{2}\int_{\mathbb{R}^{n}}|\nabla u(x)|^{2}dx, u\in{C_{0}^{\infty}(\mathbb{R}^{n} \backslash \{0\})},
    \end{equation}
    does not hold for dimension two \cite{laptev1999hardy}, and the other is that the operator \eqref{eq:ben2} in dimension two has weakly coupled eigenvalues, that is, if $\int_{\mathbb{R}^{2}}V\geq 0$, then $Neg(-\Delta-\lambda V,0)\geq 1$ for any $\lambda>0$ \cite{kovarik2011eigenvalue}.
        
    It is  shown in \cite{balinsky2001number,kovarik2011eigenvalue,laptev1999hardy} that when an additional magnetic field is introduced, the above situation improves due to the diamagnetic effects.  Balinsky et al \cite{balinsky2001number} obtained an analogy of the CLR-inequality in two dimension in case the magnetic potential is the Aharonov-Bohm type, Kovarik  \cite{kovarik2011eigenvalue}, obtained  an estimate for the number of negative eigenvalues of the magnetic Schr\"{o}dinger operator for general magnetic fields. Laptev \cite{laptev2012spectral} obtained an estimate for the number of negative eigenvalues for a magnetic Schr\"{o}dinger operator in dimension two with Ahoronov-Bohm magnetic field with an optimal constant. Barseghyan et al in \cite{barseghyan2020eigenvalue}, considered the magnetic Schr\"{o}dinger operators on the two dimensional unit disk with a radially symmetric magnetic field which explodes to infinity at the boundary, they proved a bound for the eigenvalue moments and a bound  for the number of negative eigenvalues for such operators.

    In this study, the estimate due to Balinsky et al   \cite{balinsky2001number} is extended to the case of a straight strip where the boundary operators are of the Neumann type. In particular, we derive an upper estimate for the number of discrete negative eigenvalues of the operator
    \begin{equation}\label{eq:ben8}
    H_{B,V}:=\bigg(\frac{1}{i}\nabla+A\bigg)^{2}-V \text{\;on\;} L^{2}(S),
    \end{equation}
    with magnetic field $B=\nabla \times A$, where  $A$ is the magnetic field potential with non-integer flux and $S:=\{(x_{1},x_{2})\in{\mathbb{R}^{2}}:x_{1}\in{\mathbb{R}},0<x_{2}<d\}, d>0$ is an infinite straight two-dimensional strip of constant width $d$. The boundary conditions on the lower and upper axes of $S$ are maintained to be of Neumann type. In particular, the question of whether it is possible to establish the  analogue of the CLR inequality \eqref{eq:ben6} for the operator \eqref{eq:ben8}, is considered.              
\end{section}
\begin{section}{Aim of the study}
    The aim was  to estimate the number of negative eigenvalues of the magnetic Schr\"{o}dinger operator \eqref{eq:ben8}, in a strip  in terms of the strength or regularity properties of the electric potential. 
\begin{subsection}{Objectives}
	The study was intended to achieve the following specific objectives
    \begin{enumerate}[(i).]
    \item  To obtain an analogue of the CLR  inequality when  the magnetic field  is the Aharonov-Bohm magnetic field. 
    \item To establish the relationship between the estimate and the properties of the magnetic field.
   \end{enumerate}
\end{subsection}
\end{section} 
\newpage 
\begin{section}{Justification}
    Recent developments in technology such as coherent electron beams \cite{schoenlein2019recent}, highly sensitive electron detectors and photo-lithography, ability to observe microscopic objects previously thought unobservable such as the dynamics of quantized vortices in superconductors by Lorentz microscopy all have been attributed to the Aharonov-Bohm effect. These measurement and observation techniques are expected to play a more important role in future research and development in nano-science technology, medicine (magnetic resonance imaging (MRI) machine), transportation (Maglev train), quantum computers, squids, wireless communication, generators and accelerators. The study is therefore motivated by physical models such as those involved in the theory of the surface superconductivity \cite{fournais2010spectral}
    
    Recent study in \cite{ando2016modulation} show that electric field control of magnetism in metals have attracted wide spread attention,  Yamada et al. \cite{yamada2018microscopic} pointed out that microscopic mechanism is still controversial.
    
    The need to settle such controversies and have a better understanding of the above phenomena greatly is motivating to investigate the magnetic Schr\"{o}dinger equation which gives a better  description of such phenomena via the magnetic Schr\"{o}dinger operator. In particular estimates for number of negative eigenvalues for the  magnetic Schr\"{o}dinger operator give a detailed description of the nucleation of superconductivity for super conductors of type II and an accurate estimate for the critical field, magnetic anisotrophy energy of systems and stability of matter in magnetic fields\cite{balinsky2011spectral,lieb1997stability,taylor2006scattering}. 
\end{section}
\end{chapter}

\begin{chapter}{Preliminaries, auxilliary results and Literature Review}
\begin{section}{Introduction}
	In Section \ref{Review of basic facts in spectral theory}, we review  of basic definitions and various well known results from spectral theory that was used in the study.
	
	In Section \ref{quadratic form}, we give a review on the quadratic form of self-adjoint operators on a Hilbert space.
	
	In Section \ref{Auxiliary results}, we review the auxiliary results used in the study. These include the Hardy inequality, Sturm-Liouville operator and the Bargmann inequality.
	
    Finally, in Section \ref{sectionknownresults}, we present some known results starting with the work of Melgaard and Rozenblum \cite{melgaard1996spectral} where the well-known CLR inequality for the number of negative eigenvalues of the Schr\"{o}dinger operator $-\Delta -V$ on $L^{2}(\mathbb{R}^{n})$ is carried over to a class of second order differential operators.	Generalizing the  magnetic Sch\"{o}dinger operators case with variable coefficients which may be non-smooth, unbounded and with some degeneration  is allowed.
\end{section}
\begin{section}{Review of basic facts in spectral theory}\label{Review of basic facts in spectral theory}	
	\begin{subsection}{Sobolev spaces}
		Vector space of functions that have weak derivatives are called the Sobolev spaces \cite{adams1975sobolev,leoni2017first}. The motivation behind these spaces is that the solutions of partial differential equations if they exit belong to Sobolev spaces.\\
		Let $\Omega\in{\mathbb{R}^{n}}$ be an open set, $\partial \Omega$ denotes its boundary, $u\in{C^{\prime}(\Omega)}$ and $\varphi\in{C_{0}^{\infty}(\Omega)}$.
		
		Integration by parts gives 
		\begin{equation*}
		\int_{\Omega}u\frac{\partial \varphi}{\partial x_{j}}dx=-\int_{\Omega}\frac{\partial u}{\partial x_{j}}\varphi dx
		\end{equation*}
		
		Let $u\in{C^{k}(\Omega)}$, $k=1,2,\dots,$ and let $\alpha=(\alpha_{1},\alpha_{2},\dots,\alpha_{n})\in{\mathbb{N}^{n}}$ be a multi-index such that the order of multi-index $|\alpha|=\alpha_{1}+\dots+\alpha_{n}$ is at most $k$. We denote 
		\begin{equation*}
		D^{\alpha}u=\frac{\partial^{|\alpha|}u}{\partial x_{1}^{\alpha_{1}}\dots \partial x_{n}^{\alpha_{n}}}.
		\end{equation*}
		The order of a multi-index tells the total number of differentiations. Successive integration by parts yields 
		\begin{equation*}
		\int_{\Omega}uD^{\alpha}\varphi dx=(-1)^{|\alpha|}\int_{\Omega}D^{\alpha}u\varphi dx
		\end{equation*}
		\begin{defn}
			Let $u\in{L_{loc}^{1}(\Omega)}$. For a given multi-index $\alpha$, a function $v\in{L_{loc}^{1}(\Omega)}$ is called the $\alpha^{th}$ weak derivative of $u$ if 
			\begin{equation*}
			\int_{\Omega}\varphi vdx=(-1)^{|\alpha|}\int_{\Omega}uD^{\alpha}\varphi dx \text{    \;for all\;   } \varphi \in{C_{0}^{\infty}(\Omega)}.
			\end{equation*}
		\end{defn}
		A locally integrable function is said to be $k$-times weakly differentiable if for all $|\alpha|\leq k$, the weak derivative $D^{\alpha}u$ exists.
		
		\begin{defn}
			Let $m$ be a positive integer and let $1\leq p\leq \infty$. The Sobolev space $W^{m,p}(\Omega)$ is defined by 
			\begin{equation*}
			W^{m,p}(\Omega)=\{u\in{L^{p}(\Omega)}|D^{\alpha}u\in{L^{p}(\Omega)}, \text{\;for all \;} |\alpha|\leq m \}.
			\end{equation*}
			The space $W^{m,p}(\Omega)$ is a vector space contained in $L^{p}(\Omega)$ and is endowed with the norm $\|.\|_{m,p,\Omega}$ defined as follows
			\begin{equation*}
			\|u\|_{m,p,\Omega}=\bigg(\sum_{|\alpha|\leq m}\|D^{\alpha}u\|_{L^{p}(\Omega)}^{p}(\Omega)\bigg)^{\frac{1}{p}} \qquad \text{if}\qquad 1\leq p<\infty
			\end{equation*}
			and 
			\begin{equation*}
			\|u\|_{m,\infty,\Omega}=\max_{|\alpha|\leq m}\|D^{\alpha}u\|_{L^{\infty}(\Omega)}.
			\end{equation*}
			For $m=p=2$, $W^{2,2}(\Omega)$ is a Hilbert space and $H^{2}(\Omega)$ is its standard notation.
		\end{defn}
	\end{subsection}
	\begin{subsection}{The spectrum of a self-adjoint  operator}
		Let $\mathcal{H}$ be a complex Hilbert space, $\mathcal{D}\subset{\mathcal{H}}$ a linear subset and $T:\mathcal{D}\to \mathcal{H}$ a linear map.
		
		Let $T$ be densely defined operator, that is, to say $\overline{\mathcal{D}(T)}=\mathcal{H}$. Then the adjoint operator $T^{*}$ can be constructed as follows
		\begin{equation}\label{eq:litbasic2}
		D(T^{*}):=\{v\in{\mathcal{H}}\text{  } |\text{  } \exists h\in{\mathcal{H}}:\langle Tu,v\rangle_{\mathcal{H}}=\langle u,h\rangle_{\mathcal{H}}\forall u\in{\mathcal{D}(T)}\}.
		\end{equation}
		The vector $h$ is uniquely determined by $v$, and we set $h=T^{*}v$. Thus 
		\begin{equation*}
		\langle Tu,v\rangle=\langle u,T^{*}v\rangle \quad \forall u\in{\mathcal{D}(T)}, \forall v\in{D(T^{*})}.
		\end{equation*}
		\begin{defn}
			An operator $T$ which fulfills $T^{*}=T$ is said to be {\textit{self-adjoint}}.
		\end{defn}
		An operator $T$, such that $\overline{\mathcal{D}(T)}=\mathcal{H}$ and 
		\begin{equation*}
		\langle Tu,v\rangle=\langle u,Tv\rangle, \forall u,v\in{\mathcal{D}(T)}.
		\end{equation*}
		is called {\textit{symmetric}}.
		
		$T:\mathcal{H}\to \mathcal{H}$ is said to be bounded if there exists a constant $c>0$ such that $\|Tu\|\leq c\|u\|$ for all $u\in{\mathcal{D}(T)}$. 
		
		Denote by $I$ the identity operator in the complex Hilbert space $\mathcal{H}$ and $\mathcal{B}(\mathcal{H})$ the space of bounded linear operators in $\mathcal{H}$.
		\begin{defn}[Resolvent set:  $\rho(T)$]
			Let $T$ be a closed linear operator. The complex number $\lambda$ is called a {\textit{regular point}} if the operator $T-\lambda I$ is invertible with an inverse $(T-\lambda I)^{-1}$  defined every where on $\mathcal{H}$ that belongs to $\mathcal{B}(\mathcal{H})$. The set of all regular points constitute the resolvent set of $T$ denoted by $\rho(T)$. 
		\end{defn}
		\begin{defn}[Resolvent]
			The resolvent of an operator $T$ is the operator valued function $\lambda \mapsto (T-\lambda I)^{-1}$ mapping the set $\rho(T)$ into $\mathcal{B}(\mathcal{H})$.
		\end{defn}
		\begin{defn}[Spectrum: $\sigma(T)$]
			Let $T$ be a closed linear operator. The complement of its resolvent set $\rho(T)$ is the spectrum of $T$. The spectrum is denoted by $\sigma(T)$, and  $\sigma(T)=\mathbb{C}\backslash \rho (T)$.
		\end{defn}
		If $T$ is bounded, the resolvent set $\rho (T)$ is open while the spectrum $\sigma (T)$ is closed. If $T^{*}=T$, then the spectrum of $T$ is nonempty and lies on the real axis.
		\begin{defn}[Eigenvalue]
			If there exists a non-zero vector $x\in{\mathcal{D}(T)}$ and a complex number $\lambda$ for which $Tx=\lambda x$, then $\lambda$ is called an eigenvalue of $T$ and $x$ is an associated eigenvector.
		\end{defn}
		\begin{defn}[Point spectrum and continuous spectrum]
			The set of all eigenvalues of a self-adjoint operator $T$ is called the {\textit{point spectrum}} denoted by $\sigma_{p}(T)$ while the continuous spectrum denoted by $\sigma_{c}(T)$ is defined by 
			$$\sigma_{c}(T)=\{\lambda \in{\mathbb{R}}\text{  }:\text{  }T-\lambda I \text{ is injective and has dense range, but is not surjective} \}$$.
		\end{defn}
		The union of point spectrum $\sigma_{p}(T)$ and the continuous spectrum $\sigma_{c}(T)$ of a self adjoint operator $T$ represents its spectrum $\sigma(T)$ and it is said to have a purely continuous spectrum if $\sigma_{p}(T)=\emptyset$.
		\begin{defn}
			The union of the continuous spectrum $\sigma_{c}(T)$ and the set of eigenvalues of infinite multiplicity of a self adjoint operator $T$ is called the essential spectrum.
		\end{defn}
		We denote by $\sigma_{ess}(T)$, the essential spectrum of the self adjoint operator $T$. If $\sigma_{ess}(T)=\emptyset$, then $T$ has a discrete spectrum.
	\end{subsection}
\end{section}
\begin{section}{The quadratic form of a self-adjoint operator}\label{quadratic form}
	Let $\mathcal{H}$ denote a complete separable Hilbert space. Let $\langle f,g\rangle$ and $\|f\|$, $f, g\in{\mathcal{H}}$, denote respectively, the operations of inner product and norm in $\mathcal{H}$.
	\begin{defn}
		Let a linear subspace $\mathcal{D}(q)\subset \mathcal{H}$ be fixed. A mapping $q:\mathcal{D}(q)\times \mathcal{D}(q)\to \mathbb{C}$ which satisfies
		\begin{enumerate}[(i)]
			\item $q[\alpha u+\beta v, w]=\alpha q[u,w]+\beta q[v,w]$,
			\item $q[u,\alpha v +\beta w]= \closure{\alpha}q[u,v]+\closure{\beta}q[u,w]$
		\end{enumerate}
		for all $u,v,w\in{\mathcal{D}(q)}$ is called a sesquilinear (quadratic) form. If $\mathcal{D}(q)$ is dense in $\mathcal{H}$, then $q$ is densely defined.
	\end{defn}
	The form $q$ is said to be symmetric if $q[u,v]=q\closure{[v,u]}$ and lower semi-bounded if $q[u]\geq c\|u\|^{2}$, where $c$ is a constant  and $q[u]:=q[u,u]$. A symmetric sesquilinear form is called a Hermitian form.
	\begin{defn}
		Let $T:\mathcal{H}\to \mathcal{H}$ be a linear operator. A core of $T$ is a subset $B$ of $\mathcal{D}(T)$ such that the closure of the restriction of $T$ to $B$ is $\closure{T}$.
	\end{defn}
	\begin{theorem}[{\cite[Theorem 5.37]{weidmann1980linear}}]\label{Weidmann1980them5.37}
		Let $q$ be a densely defined closed, symmetric and linear semi-bounded form in $\mathcal{H}$. Then there exists a unique self-adjoint operator $T$ on $\mathcal{H}$ such that $\mathcal{D}\subset \mathcal{D}(q)$ and $q[u,v]=\langle Tu,v\rangle $ for every $u\in{\mathcal{D}(T)}$ and $v\in{\mathcal{D}(q)}$. If $u\in{\mathcal{D}(q)}$ and $w\in{\mathcal{H}}$ such that $q[u,v]=\langle w,v\rangle$ holds for every $v$ in the core of $q$, then $u\in{\mathcal{D}(T)}$ and $Tu=w$.
	\end{theorem}
	Let $\mathcal{H}$ be a Hilbert space and let $q$ be a Hermitian form with a domain $\mathcal{D}(q)\subset \mathcal{H}$. Set 
	\begin{equation}\label{eq:sae3a}
	Neg(q):=\sup \{\dim \ell \text{ }|\text{ }q{[u]}<0, \forall u\in{ \ell \backslash \{0\}}\},
	\end{equation}
	where $\ell$ denotes a linear subspace of $\mathcal{D}(q)$.
	The number $Neg(q)$ is called the Morse index of $q$ in $\mathcal{D}(q)$. If $q$ is the quadratic form of a self adjacent operator $T$ with no essential spectrum in $(-\infty,0)$, then $Neg(q)$ is called the number of negative eigenvalues of $T$ repeated according to their multiplicity (\cite[Section 2.1]{cox2015morse}).
\end{section}
\newpage
\begin{section}{Auxiliary results}\label{Auxiliary results}
	\begin{subsection}{Hardy inequality}
		\begin{theorem}[{\cite[Theorem 327]{hardy1952inequalities}}]
			If $p>1$, $f(x)\geq 0$, and $F(x)=\int_{0}^{x}f(t)dt$, then
			\begin{equation}\label{eq:onedimhardy1}
			\int_{0}^{\infty}\bigg(\frac{F}{x}\bigg)^{p}dx\leq \bigg(\frac{p}{p-1}\bigg)^{p}\int_{0}^{\infty}f^{p}dx.
			\end{equation}
			If the right-hand side is finite, equality holds if and only if $f=0$ almost everywhere.
		\end{theorem}		
		The study of the Hardy-type inequalities related to the Magnetic Schr\"{o}dinger operator \eqref{eq:ben3}, in the recent years has risen a lot  of interest in the spectral analysis of magnetic Schr\"{o}dinger operators (see \cite[Theorems 1 and 2]{laptev1999hardy}). 
		
		The classical Hardy inequality 
		\begin{equation}\label{eq:mhardy2}
		\int_{\mathbb{R}^{n}}\frac{|u(x)|^{2}}{|x|^{2}}dx\leq \bigg(\frac{n-2}{2}\bigg)\int_{\mathbb{R}^{n}}|\nabla u(x)|^{2}dx, \text{ } u\in{C_{0}^{\infty}(\mathbb{R}^{n}\backslash\{0\})},
		\end{equation}
		does not hold for the dimension $n=2$.  Laptev and Weidl in \cite{laptev1999hardy} discovered that  by replacing the gradient $\frac{1}{i}\nabla$ with ''magnetic'' gradient $\frac{1}{i}\nabla+A$, the situation  improves, where $A$ is given by 
		\begin{equation}\label{eq:mhardy3}
		A(r,\theta)=\frac{\alpha (\theta)}{r}(\sin{\theta},-\cos{\theta}),
		\end{equation}
		where $\alpha$ is the magnetic flux.
		
		A magnetic Hardy-type inequality (see \cite[Theorem 1]{laptev1999hardy}) 
		\begin{equation}\label{eq:mhardy4}
		\int_{\mathbb{R}^{2}}\frac{|u(x)|^{2}}{1+|x|^{2}}dx\leq C \int_{\mathbb{R}^{2}}|(\frac{1}{i}\nabla+A+)u(x)|^{2}dx,
		\end{equation}
		holds true with the constant $C$ strongly depending on the magnetic field and for the Aharonov-Bohm type magnetic potential, the inequality holds with a sharp constant.
		\begin{equation}\label{eq:mhardy5}
		\int_{\mathbb{R}^{2}}\frac{|u(x)|^{2}}{|x|^{2}}dx\leq C\int_{\mathbb{R}^{2}}\bigg|\bigg(\frac{1}{i}\nabla +A\bigg)u(x)\bigg|^{2}dx, u\in{C_{0}^{\infty}(\mathbb{R}^{2}\backslash\{0\})}
		\end{equation}
		holds  true with the sharp constant $C=\bigg(\min_{k\in{\mathbb{Z}}}|k-\alpha|\bigg)^{-2}$. Here $\alpha$ is the circulation of $A$ round the origin and 
		\begin{equation*}
		\alpha=\frac{1}{2\pi}\int_{\mathbb{S}^{1}}Adx,
		\end{equation*}
		where $\mathbb{S}^{1}$ is a unit sphere.		
	\end{subsection}
	\begin{subsection}{Sturm-Liouville operator and Bargmann inequality}
		Let $\{p,q, w\}$ be the Sturm-Liouville coefficients defined on the interval $(a,b)$. Then the general Sturm-Liouville differential equation (see \cite[Section 7]{everitt2005catalogue}) is given by
		\begin{equation}\label{eq:sturm1}
		-\bigg(p(x)y^{\prime}(x)\bigg)^{\prime}+q(x)y(x)=\lambda w(x)y(x) \text{  for all $x\in{(a,b)}$},
		\end{equation}
		where $w(x)>0$ is a weight function $w:(a,b)\to \mathbb{R}$.
		
		The classical Sturm Liouville differential equation (see \cite[Section 7]{everitt2005catalogue}) is  given by 
		\begin{equation}\label{eq:sturm2}
		-y^{\prime \prime}(x)=\lambda y(x) \text{  for all $x\in{(-\infty,+\infty)}$}.
		\end{equation}
		A detail history of Sturm-Liouville theory can be found in \cite{lutzen1984sturm}.
		A sturm Liouville differential operator is  thus given by
		\begin{equation}\label{eq:sturm3}
		\tau :=-\bigg(p(x)y^{\prime}(x)\bigg)^{\prime}+q(x)y(x).
		\end{equation}
		Generally,  well-posed Sturm-Liouville boundary value problem generate self-adjoint differential operators in $L^{2}(a, b)$.
		
		\begin{lemma}[{\cite[{\bf Lemma 1.2}]{schmidt2002short}}]\label{lem3}
			Let $\tau$ be a self-adjoint Sturm Liouville  operator  on the interval $(a,b)$ and $\lambda$ be a real number. If $\tau$ has $n$ eigenvalues below $\lambda$, counting multiplicities, then there is a non-zero solution of the equation 
			\begin{equation}\label{eq:sben3b}
			\tau u=\lambda u
			\end{equation}
			with $n+1$ zeros in $(a,b)$.
		\end{lemma}  
		Let $M\in{L_{loc}^{1}(\mathbb{R})}$ and consider the following self adjoint Sturm-Liouville operator
		\begin{equation}\label{eq:proofneg1051}
		\tau=-\frac{d^{2}}{dx^{2}}+\frac{c}{x^{2}}-M(x), \text{ on $L^{2}(\mathbb{R})$},
		\end{equation}
		where $c$ is a constant.	
		\begin{lemma}\label{lembargman1}
			If the operator \eqref{eq:proofneg1051} has $n$ negative eigenvalues (counting multiplicities), then
			\begin{equation}\label{eq:bargmann3}
			\int_{\mathbb{R}}M(x)dx>n\sqrt{4c+1}.
			\end{equation}
		\end{lemma}
		\begin{proof}
			By Lemma \ref{lem3}, there is a solution of 
			\begin{equation}\label{eqbargmman1}
			-y^{\prime \prime}+\big(\frac{c}{x^{2}}-M(x)\big)y=0
			\end{equation}
			that has $n+1$ positive zeros, in a bounded interval on $\mathbb{R}$.
			
			Let $x_{1}$ and $x_{2}$ be any two consecutive zeros , then
			\begin{equation*}
			z(x):=x\frac{y^{\prime}(x)}{y(x)}, \text{   $x\in{(x_{1}, x_{2})}$}
			\end{equation*}
			is locally absolutely continuous.  
			\begin{equation*}
			\begin{aligned}
			z^{\prime}(x)&=1.\frac{y^{\prime}(x)}{y(x)}+x.\bigg(\frac{y(x).y^{\prime \prime}(x)-y^{\prime}(x)y^{\prime}(x)}{(y(x))^{2}}\bigg)\\
			&=\frac{y^{\prime}(x)}{y(x)}+x.\bigg(\frac{y^{\prime\prime}(x)}{y(x)}-\bigg(\frac{y^{\prime}(x)}{y(x)}\bigg)^{2}\bigg).
			\end{aligned}
			\end{equation*}
			Since $\frac{y^{\prime\prime}(x)}{y(x)}=\frac{c}{x^{2}}-M(x)$ and $\frac{y^{\prime}(x)}{y(x)}=\frac{1}{x}z(x)$, it follows that
			\begin{equation*}
			\begin{aligned}
			z^{\prime}(x)&=\frac{1}{x}z(x_{1})+x.\bigg(\big(\frac{c}{x^{2}}-M(x)\big)-\bigg(\frac{1}{x}z(x)\bigg)^{2}\bigg)\\
			&=\frac{-1}{x}\bigg((z(x))^{2}-z(x)\bigg)+\frac{c}{x}-xM(x)\\
			&=\frac{-1}{x}\bigg(\bigg(z(x)-\frac{1}{2}\bigg)^{2}-\frac{1}{4}\bigg)+\frac{c}{x}-xM(x).
			\end{aligned}
			\end{equation*}
			Thus, $z(x)$ satisfies
			\begin{equation*}
			z^{\prime}(x)=\frac{-1}{x}(z(x)-\frac{1}{2})^{2}+\frac{c+\frac{1}{4}}{x}-xM(x)
			\end{equation*}
			and
			\begin{equation*}
			\lim_{x\to x_{1}}z(x)=\infty, \text{ $\lim_{x\to x_{2}}z(x)=-\infty.$}
			\end{equation*}
			By setting
			\begin{equation*}
			\frac{-1}{x}(z(x)-\frac{1}{2})^{2}+\frac{c+\frac{1}{4}}{x}=0,
			\end{equation*}
			we obtain
			\begin{equation*}
			z_{\pm}(x)=\frac{1}{2}\pm \sqrt{c+\frac{1}{4}}.
			\end{equation*}
			We shall therefore have  the two points 
			\begin{equation*}
			\begin{aligned}
			x_{{+}}:=\max \{x\in{(x_{1}, x_{2})}|z(x)\geq z_{+}\}\in{(x_{1}, x_{2})},\\
			x_{{-}}:=\min \{x\in{(x_{+}, x_{2})}|z(x)\leq z_{-}\}\in{(x_{+}, x_{2})}
			\end{aligned}
			\end{equation*}
			such that $z(x_{\pm})=z_{\pm}(x)$, for  $z(x_{1})\in[z_{-},z_{+}]$ and $x\in{[x_{+}, x_{-}]}$.
			
			This implies that  $z^{\prime}(x)\geq -M(x_{1})$ for $x\in{[x_{+}, x_{-}]}$.
			
			Thus
			\begin{equation*}
			\begin{aligned}
			\int_{x_{+}}^{x_{-}}-M(x_{1})dx&\leq \int_{x_{+}}^{x_{-}}z^{\prime}(x)dx=z_{-}-z_{+}\\
			&=\bigg(\frac{1}{2}-\sqrt{c+\frac{1}{4}}\bigg)-\bigg(\frac{1}{2}+\sqrt{c+\frac{1}{4}}\bigg)\\
			&=-2\sqrt{c+\frac{1}{4}}=-\sqrt{4c+1}.
			\end{aligned}
			\end{equation*}
			So
			\begin{equation*}
			-\int_{x_{+}}^{x_{-}}M(x)dx\leq -\sqrt{4c+1}.
			\end{equation*}
			Hence 
			\begin{equation}\label{eq:bargmann2}
			\int_{-\infty}^{\infty}M(x)dx>n\sqrt{4c+1}.
			\end{equation}
		\end{proof}
	\end{subsection}
    \begin{subsection}{Diamagnetism}
    	For Schr\"{o}dinger operators, introduction of a magnetic field raises the energy.
    	\begin{theorem}[Diamagnetic inequality {\cite[Theorem 2.1.1]{fournais2010spectral}}]\label{diamagnetic inequality}
    		Let $A:\mathbb{R}^{n}\to \mathbb{R}^{n}$ be in $L_{loc}^{2}(\mathbb{R}^{2})$ and suppose that $f\in{L_{loc}^{2}(\mathbb{R}^{2})}$ is such that $(-i\nabla+A)f\in{L_{loc}^{2}(\mathbb{R}^{2})}$. Then $|f|\in{H_{loc}^{2}(\mathbb{R}^{n})}$ and 
    		\begin{equation}\label{eq:diamagnetic1}
                 \big|\nabla|f|\big|\leq \big|(-i\nabla+A)f\big|
    		\end{equation}
    		almost everywhere.
    	\end{theorem}
        The prove of Theorem \ref{diamagnetic inequality} requires us to differentiate the absolute value which we state in the proposition below.
        \begin{prop}[{\cite[Proposition 2.1.2]{fournais2010spectral}}]\label{prob diamagnetism}
        	Suppose that $f\in{L_{loc}^{1}(\mathbb{R}^{n})}$ with $\nabla f\in{L_{loc}^{1}(\mathbb{R}^{n})}$. Then also $\nabla |f|\in{L_{loc}^{1}(\mathbb{R}^{n})}$ and with the notation 
        	\begin{equation}\label{eq:diamagnetic2}
        	    Sign z=\left\{
        	    \begin{array}{rl}
        	    \frac{\bar{z}}{|z|},  \text{\;  for\;}  z\neq 0,\\
        	    0,  \text{\;  for \;} z=0,\\ 
        	    \end{array} \right.
        	\end{equation}
        	we have 
        	\begin{equation}\label{eq:diamagnetic3}
        	    \nabla|f|(x)=\Re\{Sign\big(f(x)\big)\nabla f(x)\}  \text{   for almost every  } x\in{\mathbb{R}^{n}}.
        	\end{equation}
        	In particular, 
        	\begin{equation*}
        	     \big|\nabla |f|\leq \big|\nabla f\big|
        	\end{equation*}
        	almost everywhere in $\mathbb{R}^{n}$.
        \end{prop}
        \begin{proof}
        	Suppose first that $u\in{C^{\infty}(\mathbb{R}^{n})}$ and define $|z|_{\epsilon}=\sqrt{|z|^{2}+\epsilon^{2}}-\epsilon$, for $z\in{\mathbb{C}}$ and $\epsilon>0$. We observe that
        	\begin{equation*}
        	    0\leq |z|_{\epsilon}\leq |z| \text{  and  } \lim_{\epsilon \to 0}|z|_{\epsilon}=|z|.
        	\end{equation*}
        	Then the function $|u|_{\epsilon}$ defined, for $x\in{\mathbb{R}^{n}}$ by
        	\begin{equation*}
        	    |u|_{\epsilon}(x)=|u(x)|_{\epsilon},
          	\end{equation*}
          	belongs to $C^{\infty}(\mathbb{R}^{n})$ and 
          	\begin{equation}\label{eq:diamagnetic4}
          	     \nabla|u|_{\epsilon}=\frac{\Re \big(\bar{u} \nabla u\big)}{\sqrt{|u|^{2}+\epsilon^{2}}}.
          	\end{equation}
          	Now let $f$ be as in the position and define $f_{\delta}$ as the convolution 
          	\begin{equation*}
          	     f_{\delta}=f*\rho_{\delta},
          	\end{equation*}
          	with $\rho_{\delta}$ being a standard approximation of unity for convolution. Explicitly, we take a $\rho\in{C_{0}^{\infty}(\mathbb{R}^{n})}$ with 
          	\begin{equation*}
          	    \rho \geq 0, \qquad \int_{\mathbb{R}^{n}}\rho (x)dx=1,
          	\end{equation*}
          	and define $\rho_{\delta}(x):=\delta^{-n}\rho(x/\delta)$, for $x\in{\mathbb{R}^{n}}$ as $\delta \to 0$.
          	
          	Take a test function $\varphi \in{C_{0}^{\infty}(\mathbb{R}^{n})}$. We may extract a subsequence $\{\delta_{k}\|_{k\in{\mathbb{N}}}\}$ (with $\delta_{k}\to 0$ for $k\to \infty$) such that $f_{\delta_{k}}(x)\to f(x)$ for almost every $x\in{supp \varphi}$. We restrict our attention to this subsequence. for simplicity of notation, we omit the $k$ from the notation and write $\lim_{\delta \to 0}$ instead of $\lim_{k \to \infty}$.
          	
          	We now calculate, using dominated convergence and \eqref{eq:diamagnetic4},
          	\begin{equation*}
          	\begin{aligned}
          	     \int \big(\nabla \varphi\big)|f|dx&=\lim_{\epsilon \to 0}\int \big(\nabla \varphi\big)|f|_{\epsilon}dx\\
          	     &=\lim_{\epsilon \to 0}\lim_{\delta \to 0}\int \big(\nabla \varphi\big)|f_{\delta}|_{\epsilon}dx\\
          	     &=-\lim_{\epsilon \to 0}\lim_{\delta \to 0}\int \varphi\frac{\Re \big(\bar{f_{\delta}}\nabla f_{\delta}\big)}{\sqrt{|f_{\delta}|^{2}+\epsilon^{2}}}dx.
          	\end{aligned}
          	\end{equation*}
          	Using the pointwise convergence of $f_{\delta}(x)$ and $\big\|\nabla f_{\delta} -\nabla f\big\|_{L^{1}(supp \varphi)}\to 0$, we can take the limit $\delta \to 0$ and get 
          	\begin{equation}\label{eq:diamagnetic5}
          	    \int \big(\nabla \varphi\big)|f|dx=-\lim_{\epsilon \to 0}\int \varphi\frac{\Re \big(\bar{f}\nabla f\big)}{\sqrt{|f|^{2}+\epsilon^{2}}}dx.
          	\end{equation}
          	Now, $\varphi \nabla f \in{L^{1}(\mathbb{R}^{n})}$ and $\bar{f(x)}\big(\big|f(x)\big|^{2}+\epsilon^{2}\big)^{-1/2}\to sign f(x)$ as $\epsilon \to 0$, so we get \eqref{eq:diamagnetic3} from \eqref{eq:diamagnetic5} by dominated convergence.
        \end{proof}
        \begin{proof}[Proof of Theorem \ref{diamagnetic inequality}]
        	Since $A\in{{L_{loc}^{2}(\mathbb{R}^{n})}}$ and $f\in{{L_{loc}^{2}(\mathbb{R}^{n})}}$, the assumption $\big(\nabla +iA\big)f \in{{L_{loc}^{2}(\mathbb{R}^{n})}}$ implies that $\nabla f\in{{L_{loc}^{1}(\mathbb{R}^{n})}}$. Therefore, we can use Proposition \ref{prob diamagnetism} to conclude that \eqref{eq:diamagnetic3} holds for $f$. Since $\Re \big\{ sign (f)iAf\big\}=0$, we can rewrite \eqref{eq:diamagnetic3} as 
        	\begin{equation}\label{eq:diamagnetic6}
        	     \nabla |f|=\Re \big\{sign(f)\big(\nabla +iA\big)f \big\},
        	\end{equation}
        	and therefore, since $|z|\geq |\Re (z)|$ for all $z\in{\mathbb{C}}$, we get \eqref{eq:diamagnetic1}.
        \end{proof}
        Using the diamagnetic inequality used together with the variational characterization of the ground state energy one obtains the comparison for the Dirichlet eigenvalues 
        \begin{equation*}
            \inf{\big(H_{A,V}^{D}\big)}\geq \inf{\big(H_{V}^{D}\big)},
        \end{equation*} 
        where $H_{A,V}^{D}$ and $H_{V}^{D}$ denotes the magnetic Schr\"{o}dinger operator \eqref{eq:ben4} and  Schr\"{o}dinger operator \eqref{eq:ben2} with Dirichlet boundary conditions subjected at the boundary respectively and a similar result is true in the case of Neumann boundary conditions 
        \begin{equation*}
        \inf{\big(H_{A,V}^{N}\big)}\geq \inf{\big(H_{V}^{N}\big)},
        \end{equation*} 
        where $H_{A,V}^{N}$ and $H_{V}^{N}$ denotes the magnetic Schr\"{o}dinger operator \eqref{eq:ben4} and  Schr\"{o}dinger operator \eqref{eq:ben2} with Dirichlet boundary conditions subjected at the boundary respectively \cite[Chapter 2]{fournais2010spectral}.
    \end{subsection}
\end{section}
\begin{section}{Review of known results}\label{sectionknownresults}
	 Let $\Omega$ be some open set in $\mathbb{R}^{n}$, $g(x)=\{g_{j,k}(x) \text{ $|$ } 1\leq j,k\leq n\}$  be a (possibly, complex  valued)	matrix functions defined on $\mathbb{R}^{n}$, positive definite almost everywhere such that there exists a function $\gamma (x)\geq 0, \gamma>0$ almost everywhere such that 
     \begin{equation}\label{eq:review1}
     \sum_{j=1}^{n}\sum_{k=1}^{n}g_{j,k}(x)\xi_{i}\bar{\xi_{k}}\geq \gamma (x)|\xi|^{2}.
     \end{equation}
     For any $\xi =(\xi_{1},\dots, \xi_{d})$. If $\gamma (x) \geq \gamma_{0}>0$, then the matrix is called {\textit{uniformly elliptic}}.
     
     Let $\gamma^{-1}\in{L^{p}(\Omega)}$, $p\geq \max (\frac{n}{2},1)$, $V\in{L_{loc}^{1}(\Omega)}$, $(V+\alpha)\in{L^{q}(\Omega)}$ for some $\alpha$, $\frac{1}{q}+\frac{1}{p}=\frac{2}{n}$, $q>1$.
     
     Using a standard application of the Birman-Schwinger Principle and semi group theory, Melgaard and Rozenblum \cite{melgaard1996spectral}, obtained an upper estimate for the number of negative eigenvalues for a magnetic Schr\"{o}dinger operator \eqref{eq:ben4},  where $A(x)=(A_{1}(x),\dots, A_{n}(x))$ is the magnetic potential, defined in a bounded domain $\Omega$ (for example a cube) in $\mathbb{R}^{n}$ . Suppose further that the operator is  a uniform elliptic operator (in particular, for the unit matrix $g$) and setting $p=\infty$, $q=\frac{n}{2}$ . Then
     \begin{equation}\label{eq:review8}
     	Neg(H_{A,V})\leq c\int_{\mathbb{R}^{n}}
     	  V^{n/2}(x)dx, \qquad n\geq 3,
     \end{equation}
     where $c$ is a constant \cite[Theorem 4.6]{melgaard1996spectral}. Clearly, \eqref{eq:review8} is similar to the well known CLR inequality for the electric  Schr\"{o}dinger operators.
     
     Let $\Omega$ be a space with a $\sigma$-finite measure. Any self-adjoint nonnegative operator  $T$ in  $L^{2}(\Omega)$ generates a strongly continuous contractive semi group given by
     \begin{equation}\label{eq:review9}
     Q(t)=Q_{B}(t)=e^{-tT},\qquad 0\leq t \leq \infty.
     \end{equation} 
     Denote $Q(t; x,y)=Q_{B}(t;x,y)$ and let
     \begin{equation}\label{eq:review10}
     M_{B}(t):=ess \sup_{x}\int_{\Omega}|Q_{B}(t;x,y)|^{2}dy.
     \end{equation}
          
     Let $G(z)$ be a non-negative convex function on $[0,\infty)$, polynomially growing at infinity and such that $z^{-1}G(z)$ is integrable at zero. With any such $G$ we associate another function 
     \begin{equation}\label{eq:review12}
     g(\lambda):=\int_{0}^{\infty}z^{-1}G(z)e^{-z/\lambda}dz.
     \end{equation}
     Suppose 
     \begin{equation}\label{eq:review11}
     \int_{a}^{\infty}M_{B}(t)dt<\infty,\qquad a>0.
     \end{equation}
     and let $T$ generate a positively dominated semigroup, and $V\geq 0$ be form bounded with respect $T$, with a bound smaller than 1.
     Then $T-V$ defines a self adjoint operator.
     \begin{theorem}[{\cite[Theorem 2.4]{rozenblum1998clr}}]\label{them:rozenblum1998(1)}
     	Let $B\in{\cal{P}}$ and $T\in{{\cal{P}\cal{D}}(B)}$. Suppose that $M_{B}(t)$ satisfies \eqref{eq:review11} and $M_{B}(t)={\cal{O}}(t^{\alpha})$ at zero, with some $\alpha>0$. Fix a nonnegative convex function $G$ polynomially growing at infinity and such that $G(z)=0$ near zero, put $g={\cal{L}(G)}$. Then
     	\begin{equation}\label{eq:review13}
     	Neg(T-V)\leq \frac{1}{g(1)}\int_{0}^{\infty}\int_{\Omega}\frac{1}{t}M_{B}(t)G(tV(x))dxdt,
     	\end{equation}
     	as long as the expression on the right hand side is finite.
     \end{theorem}
         
     Using the result in \eqref{eq:review13}, Rozenblum and Solomyak \cite{rozenblum1998clr} obtained the quantity $Neg(H_{A,V})$  for the magnetic Schr\"{o}dinger operator \eqref{eq:ben4} with a given magnetic vector potential $A(x)=\{A_{j}(x)\}_{1\leq j\leq n}\in{L_{loc}^{2}(\mathbb{R}^{n})}$ as 
     \begin{equation}\label{eq:review14}
     Neg(H_{A,V})\leq c(G)\int_{\mathbb{R}^{n}}V(x)^{\frac{n}{2}}dx,\qquad n\geq 3,
     \end{equation}
     where $c(G)=(2\pi)^{\frac{-n}{2}}g(1)^{-1}\int_{a}^{\infty}(t-a)t^{\frac{-n}{2}-1}dt$. For $n=3$, $a=0.25$ the constant $c(G)=0.1156$ turns out the best constant (see \cite[ Section 5.2]{blanchard1996bound}). 
     
     The proof of the magnetic CLR inequality outlined in \cite{rozenblum1998clr} gives an abstract version of Liebs estimate \cite{lieb1976bounds}, by using a direct application of the Trotter formula instead of path integrals technique  that was used by Rozenblum and Melgaard in \cite{melgaard1996spectral}, and a diamagnetic inequality, expressed in terms of positively dominated semi-groups. The constant in estimate of  Rozenblum and Solomyak \cite{rozenblum1998clr} turns out to be the best optimal. However both results obtained in  \cite{melgaard1996spectral} and in \cite{rozenblum1998clr} do not apply to the two-dimensional case.
     
     We next present the result of Rozenblum and Solomyak \cite[Theorem 1]{rozenblum1999number} where they extend results obtained in their earlier work in \cite{rozenblum1998clr} to the two-dimensional case.
     
     Let $f$ be a function in $L_{loc}^{r}(\mathbb{R}^{2})$ for some $r>1$. For an arbitrary $n=\{n_{1},n_{2}\}\in{\mathbb{Z}^{2}}$ denote $Q_{n}=(n_{1},n_{1}+1)\times (n_{2},n_{2}+1)$. Then
     \begin{equation}\label{eq:review15a}
     	S_{r}(f)=\sum_{n\in{\mathbb{Z}^{2}}}\bigg(\int_{Q_{n}}|f|^{r}dx\bigg)^{1/r}.
     \end{equation}
          
     Let $A=\{A_{1},A_{2}\}$ be a magnetic vector potential on $\mathbb{R}^{2}$, whose components are real valued functions $A_{j}=A_{j}(x_{1},x_{2})\in{L_{loc}^{2}(\mathbb{R}^{2})}$, $j=1,2$, $A\in{L_{loc}^{2}(\mathbb{R}^{2})}$, $V\in{L_{loc}^{1}(\mathbb{R}^{2})}$ and $S_{r}(V)<\infty$ for some $r>1$.
     
     By studying the discrete spectrum in gaps for perturbation of the magnetic Schr\"{o}dinger operator  \eqref{eq:ben4} on $\mathbb{R}^{2}$,  Rozenblum and Solomyak \cite{rozenblum1999number} obtained an estimate from above for the number of negative eigenvalues for the operator \eqref{eq:ben4} when $n=2$ given below
     \begin{theorem}[{\cite[Theorem 1.]{rozenblum1999number}}]\label{themrozenblum19991}
     	Let $A\in{L_{loc}^{2}(\mathbb{R}^{2})}$, $V\in{L_{loc}^{1}(\mathbb{R}^{2})}$ and $S_{r}(V_{+})<\infty$ for some $r>1$. Then the operator \eqref{eq:ben4} when $n=2$ is well defined as the form sum, its negative spectrum is discrete, and the following inequality is satisfied 
     	\begin{equation}\label{eq:review16}
     	      Neg(H_{A,V})\leq C_{r}S_{r}(V),
        \end{equation}
        The constant $C_{r}$ does not depend on $A$ and $V$.
    \end{theorem}
      
     Using a Bargmann inequality \cite{bargmann1952number} and the magnetic Hardy inequality obtained by Laptev and Weidl \cite{laptev1999hardy}, Balinsky, Evans and Lewis in \cite{balinsky2001number}, obtained an estimate  from above for the number negative eigenvalues for a rotationally symmetric two-dimensional Schr\"{o}dinger operator with an Aharonov-Bohm magnetic potential with a non-integer magnetic flux. 
     
     They considered the Aharonov-Bohm magnetic potential defined by 
     \begin{equation}\label{eq:review18}
     A(r,\theta):=\frac{\Phi (\theta)}{r}(\sin{\theta},-\cos{\theta}),\qquad \Phi \in{L^{\infty}(\mathbb{S}^{1})},
     \end{equation}
     where $\mathbb{S}^{1}=\{(r,\theta): r=1,0\leq \theta\leq 2\pi \}$, is a unit sphere and the magnetic Schr\"{o}dinger operator  \eqref{eq:ben4} on $L^{2}(\mathbb{R}^{2})$. 
     \begin{theorem}[{\cite[Theorem 1.3]{balinsky2001number}}]\label{thembalinsky2001}
     	Let $A$ be given by \eqref{eq:review18}, $V\in{L_{loc}^{1}(\mathbb{R}^{2}\backslash \{0\})}$ and $V\in{Y=L^{1}(\mathbb{R}^{+},L^{\infty}(\mathbb{S}^{1}),rdr)}$, and suppose that the magnetic flux $\tilde{\Phi}$ is not an integer. Then 
     	\begin{equation}\label{eq:review20}
     	Neg(H_{A,V})\leq {{\sum_{m\in{\mathbb{Z}}}}}^{\prime}\frac{1}{2|m+\tilde{\Phi}|}\int_{0}^{\infty}ess\sup_{\theta\in{\mathbb{S}^{1}}}|V(r,\theta)|rdr,
     	\end{equation}
     	where $\sum^{\prime}$ means that all the summands less than one are omitted since $Neg(H_{A,V})$ is an integer.
     \end{theorem}
     Before we proof the result above, we state the following two preliminary results
     \begin{lem}[{\cite[Lemma 5.6.2]{balinsky2015analysis}}]\label{lemessential}
     	Let $A$ be defined as in \eqref{eq:review18}. Then $\sigma_{ess}(H_{B})=[0,\infty)$.
     \end{lem}
     \begin{theorem}[{\cite[Theorem 1.2]{balinsky2001number}}]\label{balinsky2015}
     	Suppose that  $\tilde{\Phi}$ is not an integer, and let $V\in{Y=L^{1}(\mathbb{R}^{+},L^{\infty}(\mathbb{S}^{1}),rdr)}$. Then $V$ is form bounded relative to $H_{B}$ with $H_{B}$-bound zero, and hence $H_{A,V}$ is defined as a form sum with form domain $\mathcal{Q}(H_{B})=H_{A}^{1}$. Moreover, the essential spectra of $H_{A,V}$  and $H_{B}$ coincide.
     \end{theorem}
     \begin{proof}[Proof of Theorem \ref{thembalinsky2001}]
     	The operator $H_{A,V}$ is a self-adjoint operator associated with the lower semi-bounded quadratic form
     	\begin{equation*}
     	      h[u]:=\int_{\mathbb{R}^{2}}\bigg(|(\nabla +i A)u|^{2}-V|u|^{2}\bigg)dx, \qquad u\in{C_{0}^{\infty}(\mathbb{R}^{2}\backslash \{0\})}.
     	\end{equation*}
     	The form domain of $H_{A,V}$ is the domain of the closure of $h$, and it coincides with the Sobolev space $H_{A}^{1}$  defined as the completion of $C_{0}^{\infty}(\mathbb{R}^{2}\backslash \{0\})$ with respect to the norn
     	\begin{equation*}
     	     \|u\|_{H_{A}^{1}}^{2}=\|(\nabla +i A)u\|^{2}+\|u\|^{2},
     	\end{equation*}
     	where $\|.\|$ denotes the $L^{2}(\mathbb{R}^{2})$ norm. 
     	Let 
     	\begin{equation}\label{eq:balinsky1}
     	      Q(r):=\|V(r,.)\|_{L^{\infty}(\mathbb{S}^{1})},
     	\end{equation}     	
        so that 
        \begin{equation}\label{eq:balinsky2}
               \|Q\|_{L^{1}(\mathbb{R}^{+},rdr)}=\|Q\|_{L^{1}(\mathbb{R}^{+},L^{\infty}(\mathbb{S}^{1}),rdr)}=\|V\|_{L^{1}(\mathbb{R}^{+},L^{\infty}(\mathbb{S}^{1}),rdr)}<\infty.
        \end{equation}
        Thus by Theorem \ref{balinsky2015} (see also\cite[Lemma, 5.6.1]{balinsky2015analysis}), $H_{A,Q}:=\big(\frac{1}{i}\nabla +A\big)^{2}-Q$ is lower semi-bounded, self-adjoint and has essential spectrum $[0,\infty)$.
        
        Defining $Q$ as in \eqref{eq:balinsky1} implies that $Q\geq V$ and so $H_{A,V}\geq H_{A,Q}$. We obtain that
        \begin{equation}\label{eq:balinsky3}
              Neg(H_{A,V})\leq Neg(H_{A,Q}).
        \end{equation}
        It is enough therefore to prove the theorem with $V$ replaced by $Q$. By \cite[Section 2]{adami1998aharonov} and \cite{laptev1999hardy}, we have that 
        \begin{equation}\label{eq:balinsky3a}
             L^{2}(\mathbb{R}^{2})=L^{2}(\mathbb{R}^{+};rdr)\otimes L^{2}(\mathbb{S}^{1})=\oplus_{m\in{\mathbb{Z}}}\big\{L^{2}(\mathbb{R}^{+};rdr)\otimes \big[ \frac{e^{im \Phi}}{\sqrt{2\pi}}\big] \big\},
        \end{equation}
        where $\big[.\big]$ denotes the linear span such that 
        \begin{equation}\label{eq:balinsky3b}
            H_{B}=\oplus_{m\in{\mathbb{Z}}}\big\{ D_{M}\otimes 1_{m}\big\},
        \end{equation}
        where $$D_{m}=-\frac{1}{r}\frac{d}{dr}(r\frac{d}{dr})+(\frac{m+\Phi}{r})^{2}$$ is the (Friedrichs) operator in $L^{2}((0,\infty),rdr)$ associated with the quadratic form 
        \begin{equation}\label{eq:balinsky6}
             h_{m}[u]:=\int_{0}^{\infty}\bigg(|u_{m}^{\prime}(r)|^{2}+\frac{(m+\Phi)^{2}}{r^{2}}|u_{m}(r)|^{2}\bigg)rdr.
        \end{equation}
        For any $u$ in $H_{A}^{1}$ the form domain of $H_{A,Q}$, we have
        \begin{equation}\label{eq:balinsky6a}
             u(r,\theta)=\sum_{m\in{\mathbb{Z}}}u_{m}(r)\big\{ \frac{e^{im \Phi}}{\sqrt{2\pi}} \big\},
        \end{equation} 
        where the Fourier coefficients $u_{m}(r)$,$m\in{\mathbb{Z}}$ are given by
        \begin{equation*}
             u_{m}(r)=\int_{0}^{2\pi}u(r,\theta)\big\{ \frac{e^{im \Phi}}{\sqrt{2\pi}} \big\}d\theta
        \end{equation*}
        and
        \begin{equation}\label{eq:balinsky4}
              \int_{\mathbb{R}^{2}}\bigg(|(\nabla +i A)u|^{2}-Q|u|^{2}\bigg)dx=\sum_{m\in{\mathbb{Z}}}\int_{0}^{\infty}\bigg(|u_{m}^{\prime}|^{2}+\frac{(m+\Phi)^{2}}{r^{2}}|u_{m}|^{2}-Q|u_{m}|^{2}\bigg)rdr.
        \end{equation}
        It follows that 
        \begin{equation}\label{eq:balinsky5}
              H_{A,Q}=\oplus_{m\in{\mathbb{Z}}}\big\{(D_{m}-Q)\otimes 1_{m} \big\},
        \end{equation}
        The negative spectrum of $H_{A,Q}$ is the aggregated of the negative eigenvalues of the operator $D_{m}-Q$. To complete the proof, we use the Bargmann estimate \cite[Estimate (3)]{bargmann1952number}
        \begin{equation}\label{eq:balinsky7}
        		Neg(D_{m}-Q)\leq \frac{1}{2|m+\Phi|}\int_{0}^{\infty}Q(r)rdr.
        \end{equation}
        By \eqref{eq:balinsky5}, \eqref{eq:balinsky7} yields 
        \begin{equation}\label{eq:balinsky8}
              Neg(H_{A,Q})\leq {\sum_{m\in{\mathbb{Z}}}}^{\prime}\frac{1}{2|m+\Phi|}\int_{0}^{\infty}Q(r)rdr.
        \end{equation}
        Hence \eqref{eq:review20} follows by replacing $Q$ by $V$ in \eqref{eq:balinsky8}.
     \end{proof}
     The above estimate  fails to hold as soon as the magnetic field is more regular and in addition the estimate obtained is not optimal,  where optimality means that $Neg(H_{A,V})$ has to be infinite if the right hand side of above estimate \eqref{eq:review20} is infinite, and vice versa.     
       
     Considering the Hamiltonian of a charged quantum particle in $\mathbb{R}^{2}$ interacting with a magnetic field $B=curl A$ given by \eqref{eq:ben3} in $L^{2}(\mathbb{R}^{2})$, Kovarik in \cite{kovarik2011eigenvalue} established an upper bound $Neg(H_{A,V})$ for the magnetic Schr\"{o}dinger operator with a reasonably large class of magnetic fields extending the result obtained by Balinsky et al \cite{balinsky2001number} to fields that are more regular. Kovarik achieved this by using the Beurling-Deny criteria (see \cite[Corrollary 2.18]{ouhabaz2009analysis}, or \cite{arendt2002semigroups})  and the Trotter formula to obtain a bound for the heat kernel generated by the Schr\"{o}dinger semigroup. The bound for the heat kernel obtained, a diamagnetic inequality \cite{hundertmark2004diamagnetic}, a magnetic Hardy inequality (see\cite[Theorem 1]{laptev1999hardy}) and a generalization of Liebs inequality (see \cite{lieb1976bounds} or \cite[Theorem 2.5]{rozenblum1998clr}), was used in \cite{kovarik2011eigenvalue} to obtain the following estimate from above for the number of negative eigenvalues for the operator $H_{A,V}$, $V\geq 0$.
     \begin{assumption}\label{assump1}
     	 Assume that there exists $\epsilon\in{(0,\frac{1}{2})}$, $A=A(\epsilon)$ and a finite number $M$ of open intervals $I_{j}=(\alpha_{j},\beta_{j})$, such that 
     	\begin{equation*}
     	\begin{split}
     	\{r>0:\min_{k\in{\mathbb{Z}}}|k-\phi(r)|<\epsilon\}\subset \cup_{j=1}^{M}I_{j},\\
     	\beta_{j-1}<\alpha_{j}<\beta_{j},\qquad j=1,\dots,M,\\
     	|I_{j}|\leq A\min\{1+\alpha_{j},\alpha_{j}-\beta_{j-1},\alpha_{j+1}-\beta_{j}\}, \qquad 1\leq j\leq M.
     	\end{split}
     	\end{equation*}
     \end{assumption}
     \begin{theorem}[{\cite[Theorem 3.4]{kovarik2011eigenvalue}}]\label{themkovarik2011}
     	Assume that $A\in{L_{loc}^{2}(\mathbb{R}^{2})}$ generates a magnetic field $B$ which satisfies  the above assumption. Let $0\leq V\in{L_{loc}^{2}(\mathbb{R}^{2})}\cap L^{1+a}(\mathbb{R}^{2}, (1+|x|)^{2a}dx)$ for some $a>0$. Then 
     	\begin{equation}\label{eq:review23}
     	Neg(H_{A,V})\leq c(B,a)\int_{\mathbb{R}^{2}}V(x)^{1+a}(1+|x|)^{2a}dx.
     	\end{equation}
     \end{theorem}
     \begin{proof}
     	Let $0\leq \rho \leq 1, \rho \neq 0$ be a radial function from $C^{1}(\mathbb{R}^{2})$ with support in $\mathcal{B}_{1}=\{x\in{\mathbb{R}^{2}}:|x|<1\}$. Introduce a family of potential functions $U_{\beta}$ given by
     	\begin{equation}\label{eq:result2neg2}
     	U_{\beta}(x)=U_{\beta}(|x|)=\left\{
     	\begin{array}{rl}
     	\beta^{2}  \text{\;if\;}  |x|\leq 1,\\
     	\frac{\beta^{2}}{|x|^{2}}  \text{\;if \;} |x|>1,\\ 
     	\end{array} \right.
     	\end{equation} 
     	where $\beta>0$ and $U_{0}(x)=U_{0}(|x|)=\rho(|x|)$.
     	
     	Next we define Schr\"{o}dinger operators 
     	$$\mathcal{ A}_{\beta}=-\Delta+U_{\beta} \text{   in $L^{2}(\mathbb{R}^{2})$ }.$$
     	The operators $\mathcal{A}_{\beta}$ generates contraction semigroups $e^{-t\mathcal{A}_{\beta}}$ on $L^{2}(\mathbb{R}^{2})$ with almost everywhere positive integral kernels $K_{\beta}(t,x,y)=e^{-t\mathcal{A}_{\beta}}(x,y)$ where $t>0$.
     	For almost every $x,y\in{\mathbb{R}^{2}}$ and all $t>0$ we have
     	\begin{equation}\label{eq:result2neg3}
     		K_{\beta}(t,x,y)=e^{-t\mathcal{A}_{\beta}}(x,y)\leq C\min \{\frac{1}{t},\frac{(1+|x|)^{2\beta}}{t^{(1+\beta)}} \}, \beta>0,
     	\end{equation} 
     	for some constant $C$ \cite[Lemma 5.1]{kovarik2011eigenvalue}.
     	
     	Fix $a>0$, $C_{B}>0$ and choose $\epsilon>0$ such that 
     	\begin{equation}\label{eq:proofthemneg1}
     	a^{2}=\bigg(\frac{1-\epsilon}{\epsilon}\bigg)C_{B}.
     	\end{equation}
     	Let $\epsilon\in{(0,1)}$. By the Hardy type inequality \cite[Theorem 1]{laptev1999hardy} and the variational principle, we get  
     	\begin{equation}\label{eq:proofthemneg2}
     	Neg(H_{A,V})\leq Neg(H_{B}+U_{a}-\epsilon^{-1}V)     
     	\end{equation}
     	for each $a>0$, the operator $H_{B}+U_{a}$ generates a contractive semigroup $e^{-s(H_{B}+U_{a})} \text{\;in\;} L^{2}(\mathbb{R}^{2})$ where $s>0$. Let
     	\begin{equation}\label{eq:proofthemneg5}
     	K_{a}(s,x,y)=e^{-s(H_{B}+U_{a})}(x,y), \text{ $s>0$, }\qquad x,y\in{\mathbb{R}^{2}}
     	\end{equation}
     	be its integral kernel. By the diamagnetic inequality  \eqref{eq:diamagnetic1} (see \cite[Theorem 2.1.1]{fournais2010spectral} or \cite{hundertmark2004diamagnetic}) we have 
     	\begin{equation}\label{eq:proofthemneg6}
     	|K_{a}(s,x,y)|\leq K_{a}(s,x,y),\qquad a>0 \text{ almost everywhere}   \qquad x,y\in{\mathbb{R}^{2}}.
     	\end{equation}
     	By Lieb's inequality (see \cite{lieb1976bounds} or \cite[Theorem 2.5]{rozenblum1998clr}) we obtain
     	\begin{equation}\label{eq:proofthemneg7}
     	\begin{aligned}
     	Neg(H_{B}+U_{a}-\epsilon^{-1}V)&\leq C_{B}\int_{0}^{\infty}\int_{\mathbb{R}^{2}}\frac{1}{t}K_{a}(t,x,y)(tV(x)-1)dxdt\\
     	&\leq C_{a}\int_{\mathbb{R}^{2}}\int_{\frac{1}{V(x)}}^{\infty}K_{a}(t,x,y)(V(x))dtdx.
     	\end{aligned}
     	\end{equation}
     	Setting 
     	\begin{equation}\label{eq:proofthemneg8}
     	t_{a}(x)=e+\frac{1}{V(x)},
     	\end{equation}	
     	and perform the integration with respect to $t$ using the estimate \eqref{eq:result2neg3}, we shall  obtain that 
     	\begin{equation}\label{eq:proofthemneg10}
     	Neg(H_{B}+U_{a}-\epsilon^{-1}V)\leq C_{B}\int_{\mathbb{R}^{2}}\bigg\{\int_{\frac{1}{V(x)}}^{\infty}K_{a}(t,x,y)(V(x)dt\bigg\}dx
     	\end{equation}\label{eq:proofthemneg11}
     	Now let's consider  evaluation of the integral $\int_{\frac{1}{V(x)}}^{\infty}K_{a}(t,x,y)(V(x))dt $ using the estimate, 
     	\begin{equation}\label{eq:proofthemneg12}
     	K_{a}(t,x,y)=e^{-t\mathcal{A}_{a}}(x,y)\leq C\min \bigg\{\frac{1}{t},\frac{(1+|x|)^{2a}}{t^{(1+a)}} \bigg\}, a>0.
     	\end{equation}
     	and from \eqref{eq:proofthemneg8},  $t_{a}-e=\frac{1}{V}$ and  we  obtain 
     	\begin{equation}\label{eq:proofthemneg13}
     	\begin{aligned}
     	\int_{\frac{1}{V(x)}}^{\infty}K_{a}(t,x,y)V(x)dt&=V(x)\int_{t_{a}-e}^{\infty}K_{a}(t,x,y)dt\\
     	&=V(x)\int_{t_{a}-e}^{\infty}C\min \bigg\{\frac{1}{t},\frac{(1+|x|)^{2a}}{t^{(1+a)}} \bigg\}dt.
     	\end{aligned}
     	\end{equation}
     	Consider the integral
     	\begin{equation*}
     	\begin{aligned}
     	\int_{t_{a}-e}^{\infty}\frac{1}{t}dt&=\ln{t}|_{t_{a}-e}^{\infty}\\
     	&=\ln{\infty}-\ln({t_{a}-e})\\
     	&=\infty
     	\end{aligned}
     	\end{equation*}
     	and 
     	\begin{equation*}
     	\begin{aligned}
     	\int_{t_{a}-e}^{\infty}\frac{(1+|x|)^{2a}}{t^{(a+1)}}dt&=(1+|x|)^{2a}\int_{t_{a}-e}^{\infty}\frac{1}{t^{(a+1)}}dt\\
     	&=(1+|x|)^{2a}\int_{t_{a}-e}^{\infty}t^{-(a+1)}dt=(1+|x|)^{2a}\frac{t^{-a}}{-a}|_{t_{a}-e}^{\infty}\\
     	&=\frac{(1+|x|)^{2a}t^{-a}}{-a}|_{t_{a}-e}^{\infty}\\
     	&=\bigg( \frac{\infty^{-a}}{-a}(1+|x|)^{2a}\bigg)-\bigg(\frac{(1+|x|)^{2a}(t_{a}-e)^{-a}}{-a}\bigg)\\
     	&=\frac{-1}{a.\infty^{a}}(1+|x|)^{2a}+\frac{(1+|x|)^{2a}}{a(t_{a}-e)^{a}}=0+\frac{(1+|x|)^{2a}}{a(V(x))^{-a}}\\
     	\int_{t_{a}-e}^{\infty}\frac{(1+|x|)^{2a}}{t^{(a+1)}}dt&=\frac{(1+|x|)^{2a}}{a(V(x))^{-a}}.	
     	\end{aligned}
     	\end{equation*} 
     	It follows therefore that,
     	\begin{equation}\label{eq:proofthemneg14}
     	\begin{aligned}
     	\int_{\frac{1}{V(x)}}^{\infty}K_{a}(t,x,y)V(x)dt&=V(x).C\min \bigg\{\infty,\frac{(1+|x|)^{2a}}{a(V(x))^{-a}} \bigg\}\\
     	&=\frac{C.V(x)(1+|x|)^{2a}}{a(V(x))^{-a}}.
     	\end{aligned}
     	\end{equation}
     	Next we obtain
     	\begin{equation}\label{eq:proofthemneg15}
     	\begin{aligned}
     	Neg(H_{B}+U_{a}-\epsilon^{-1}V)&\leq C_{a}\int_{\mathbb{R}^{2}}\frac{C.V(x)(1+|x|)^{2a}}{a(V(x))^{-a}}dx.\\
     	&=\frac{C.C_{B}}{a}\int_{\mathbb{R}^{2}}(1+|x|)^{2a}(V(x))^{a+1}dx.
     	\end{aligned}
     	\end{equation}
     	Since $Neg(H_{A,V})\leq Neg (H_{B}+U_{a}-\epsilon^{-1}V)$. We shall obtain that,
     	\begin{equation}\label{eq:proofthemneg16}
     	Neg(H_{A,V})\leq C(B,a)\int_{\mathbb{R}^{2}}(1+|x|)^{2a}(V(x))^{a+1}dx.
     	\end{equation}
     	Hence 
     	\begin{equation}\label{eq:proofthemneg17}
     	Neg(H_{A,V})\leq C(B,a)\int_{\mathbb{R}^{2}}(1+|x|)^{2a}(V(x))^{a+1}dx.
     	\end{equation}
     \end{proof}
     
     An optimal estimate for the number of negative eigenvalues for the magnetic Schr\"{o}dinger operator for two dimensional case was obtained by Laptev in \cite{laptev2012spectral}.
     
     Consider in $L^{2}(\mathbb{R}^{2})$ the operator $H_{\beta}$ 
     \begin{equation}\label{eq:review24}
     H_{\beta}=(i\nabla+{\cal{A}}_{\beta})^{2}+V(x),
     \end{equation}
     where $V(x)=V(|x|)$ and the magnetic potential is given by 
     \begin{equation}\label{eq:review25}
     {A}_{\beta}=\beta\big(\frac{x_{2}}{|x|^{2}},-\frac{x_{1}}{|x|^{2}}\big), \qquad \beta\in{(0,1)}.
     \end{equation}
     By decomposing the operator $H_{\beta}$ to a one-dimensional case, let $\{-\lambda_{k}\}$ be the negative eigenvalues of the one-dimensional Schr\"{o}dinger operator obtained.
      
     Laptev in \cite{laptev2012spectral} proved the estimate by A.A. Balinsky,W.D.Evans and R.T.Lewis in \cite{balinsky2001number} with optimal constant by using the sharp inequality of Hundertmark-Lieb-Thomas (see \cite[Section 2, Theorem 1]{hundertmark2002sharp}, \cite[Theorem 5.1]{truc2012eigenvalue}), given by
     \begin{equation}\label{eq:review28}
     Neg(H_{\beta})\leq \frac{R(\beta)}{4\pi}\int_{\mathbb{R}^{2}}V(x)dx.
     \end{equation}
          
     Let $A=\sum_{j=1}^{2}A_{j}dx_{j}$ be the magnetic flux associated to the magnetic field $B=b(x)dx_{1}\wedge dx_{2}$, where $b(x)=\partial_{1}A_{2}(x)-\partial_{2}A_{1}(x)$, and $V$ a scalar potential on the unit disk
     \begin{equation}\label{eq:review29}
     \Omega:=\{x=(x_{1},x_{2})\in{\mathbb{R}^{2}}\text{ }|\text{ }x_{1}^{2}+x_{2}^{2}=r<1\}.
     \end{equation}
     such that $k=\inf_{x\in{\Omega}}b(x)>0$ and $b(x)\to +\infty$ as $x$ approaches the boundary. Then the two-dimensional magnetic Schr\"{o}dinger operator 
     \begin{equation}\label{eq:review29a}
     H_{A}^{D}=-\sum_{j=1}^{2}(\frac{\partial}{\partial x_{j}}-iA_{j})^{2}, \text{ in } L^{2}(\Omega),
     \end{equation}
     with a non-vanishing radial magnetic field and Dirichlet conditions at the boundary of the unit disk. Truc in \cite[Theorem 2.1]{truc2012eigenvalue} obtained the following estimate
     \begin{equation}\label{eq:review32}
     Neg(H_{A}^{D}-V)\leq \frac{1}{\sqrt{1-\alpha}}\int_{0}^{1}\big((\frac{1}{\alpha}-1)\frac{A^{2}(r)}{r^{2}}+V(r)\big)rdr+2\int_{0}^{1}\big(1+|\log{(r\sqrt{k})|}\big)V(r)rdr,
     \end{equation}
     for any $\alpha\in{(0,1)}$, $A(r)$ is the magnetic flux of the magnetic field through the disk of radius $r$ centered at the origin.
     
     The main tools used by Truc in \cite{truc2012eigenvalue}, were the sharp inequality of the Hundertmark-Lieb-Thomas (see \cite[Section 2, Theorem 1]{hundertmark2002sharp}) and the Birman Schwinger Principle. Truc in \cite{truc2012eigenvalue}, found out that the total flux is not necessary finite and the upper bound involves explicitly the square of the magnetic potential as opposed to the result obtained by Kovarik in \cite{kovarik2011eigenvalue} where  the total magnetic flux has to be finite and the dependence on the magnetic field is not explicit even in the radial case.
     
     Let $\Omega$ be an open set in $\mathbb{R}^{2}$. Let $A=(A_{1},A_{2}):\Omega\to \mathbb{R}^{2}$ be the magnetic potential, $V:\Omega \to \mathbb{R}$ be the electric potential and the magnetic field $B=\frac{\partial A_{2}}{\partial x_{1}}-\frac{\partial A_{1}}{\partial x_{2}}$. Then the two-dimensional magnetic Schr\"{o}dinger operator is given by 
     \begin{equation}\label{eq:review33a}
     H_{\Omega}(A,V)=(i\nabla +A)^{2}-V, \text{ on } L^{2}(\Omega).
     \end{equation}
     On $\partial \Omega$, Dirichlet boundary conditions are subjected. Let $\tilde{V}:\mathbb{R}_{+}\to \mathbb{R}$ be given by 
     \begin{equation}\label{eq:review33}
     \tilde{V}(r):=ess \sup_{\theta\in{[0,2\pi)}}V(r,\theta).
     \end{equation}
     Barseghyan and Schneider in  \cite{barseghyan2020eigenvalue} obtained the estimate 
     \begin{equation}\label{eq:review34}
     Neg(H_{\Omega}(A,V))\leq 1+c_{1}(B)\int_{0}^{1}\tilde{V}_{+}(r)(1+|\ln{r}|)rdr,
     \end{equation}
     for some positive constant $c_{1}=c_{1}(B)$ depending on the magnetic field $B$ (see \cite[Theorem 2.1]{barseghyan2020eigenvalue}).
      
     The knowledge gap in this study is the lack of optimal estimate for general magnetic field and limited literature on the estimates of number of negative eigenvalues for a magnetic Schr\"{o}dinger operator in a wave guide. Briet et al.\cite{briet2007spectral} gives spectral properties of a magnetic Quantum Hamiltonian on a strip not an estimate for the number of negative eigenvalues unlike the case for the electric Schr\"{o}dinger operator where estimates for the number of negative eigenvalues in the strip have been obtained  \cite{karuhanga2018eigenvalue,karuhanga2019discrete}.
\end{section}
\end{chapter}

\begin{chapter}{Magnetic Schr\"{o}dinger operator in a strip with Aharonov-Bohm magnetic field}
	
\begin{section}{Introduction}
	In this chapter, an upper estimate for the number of negative eigenvalues for the magnetic Schr\"{o}dinger operator with Aharonov-Bohm type magnetic field in a strip is presented. The estimate is given in terms of regularity properties of the electric potential.
		
	We consider the  magnetic Schr\"{o}dinger operator on $L^{2}(S)$ 
	\begin{equation}\label{eq:result1clr1}
	H_{B,V}=\bigg( \frac{1}{i}\nabla +A\bigg)^{2}-V, \qquad V\geq 0
	\end{equation}
	where $V$ is the electric potential, $S:=\{(x_{1},x_{2})\in{\mathbb{R}^{2}}:x_{1}\in{\mathbb{R}},0<x_{2}<d\}, d>0$ is an infinite straight strip,  and the $A$ is magnetic potential of Aharonov-Bohm field subject to Neumann boundary conditions.
	
	Without loss of generality, assume that  $V$ is non-negative, $A\in{L_{loc}^{2}(\mathbb{R}^{2})}$ and  $V\in{L_{loc}^{1}(\mathbb{R}^{2})}$. The magnetic potential $A$ is given by 
	\begin{equation}\label{eq:proofneg2}
	A=\Psi\bigg(\frac{-x_{2}}{x_{1}^{2}+x_{2}^{2}},\frac{x_{1}}{x_{1}^{2}+x_{2}^{2}}\bigg),
	\end{equation}
	where
	\begin{equation}\label{eq:proofneg3}
	\Psi(x_{1})=\frac{1}{d}\int_{0}^{d}A(x_{1},x_{2})dx_{2},
	\end{equation}
	is the magnetic flux through the strip.
		
	By the construction above, $A=A(x_{1},x_{2})=(a_{1}(x_{1},x_{2}),a_{2}(x_{1},x_{2}))$, where 
	\begin{equation*}
	a_{1}(x_{1},x_{2})=\Psi.\frac{-x_{2}}{x_{1}^{2}+x_{2}^{2}} \text{\; and\;} a_{2}(x_{1},x_{2})=\Psi.\frac{x_{1}}{x_{1}^{2}+x_{2}^{2}}
	\end{equation*}
	Now by definition,  $B(x_{1},x_{2})=curl A=\frac{\partial}{\partial{ x_{1}}} a_{2}(x_{1},x_{2})-\frac{\partial}{\partial{ x_{2}}}  a_{1}(x_{1},x_{2})$ 
	where
	\begin{equation*}
	\frac{\partial}{\partial{ x_{1}}} a_{2}(x_{1},x_{2}) =\frac{\partial}{\partial x_{1}}\bigg(\Psi.\frac{x_{1}}{x_{1}^{2}+x_{2}^{2}}\bigg)=\Psi.\frac{(x_{1}^{2}+x_{2}^{2}).1-x_{1}.2x_{1}}{(x_{1}^{2}+x_{2}^{2})^{2}}=\Psi.\frac{x_{2}^{2}-x_{1}^{2}}{(x_{1}^{2}+x_{2}^{2})^{2}} \text{ and}
	\end{equation*}
	\begin{equation*}
	\frac{\partial}{\partial{ x_{2}}}  a_{1}(x_{1},x_{2})=\frac{\partial}{\partial x_{1}}\bigg(\Psi.\frac{-x_{2}}{x_{1}^{2}+x_{2}^{2}}\bigg)=-\Psi.\bigg(\frac{(x_{1}^{2}+x_{2}^{2}).1-x_{2}.2x_{2}}{(x_{1}^{2}+x_{2}^{2})^{2}}\bigg)=\Psi.\frac{x_{2}^{2}-x_{1}^{2}}{(x_{1}^{2}+x_{2}^{2})^{2}}
	\end{equation*}
	It therefore follows that
	$curlA=\Psi.\bigg(\frac{x_{2}^{2}-x_{1}^{2}}{(x_{1}^{2}+x_{2}^{2})^{2}}-\frac{x_{2}^{2}-x_{1}^{2}}{(x_{1}^{2}+x_{2}^{2})^{2}}\bigg)=0$ thus, $B(x_{1},x_{2})=0$.\\
	In addition,  
	\begin{equation*}
	A.x=\begin{pmatrix}	\Psi\frac{-x_{2}}{x_{1}^{2}+x_{2}^{2}}\\\Psi\frac{x_{1}}{x_{1}^{2}+x_{2}^{2}}	\end{pmatrix}.\begin{pmatrix}x_{1}\\x_{2}\end{pmatrix}=\Psi\frac{-x_{1}x_{2}}{x_{1}^{2}+x_{2}^{2}}+\Psi\frac{x_{1}x_{2}}{x_{1}^{2}+x_{2}^{2}}=0.
	\end{equation*} 
	Hence $A$ is a magnetic vector potential of the Aharonov-Bohm type (see \cite[Section 5.4]{balinsky2015analysis}).
		
	With $\nabla_{A}:=(\frac{1}{i}\nabla +A)$ and $\Delta_{A}:=(\frac{1}{i}\nabla +A)^{2}$, one can perform a gauge transformation, that is conjugation by $e^{i\Psi}$. Then, since $e^{-i\Psi}\nabla_{A}e^{i\Psi}=\nabla_{A}+\nabla \Psi$, we get the unitary equivalence of $H_{B,V}:=(\frac{1}{i}\nabla +A)^{2}-V$ and $H_{B+\nabla \Psi,V}:=(\frac{1}{i}\nabla +A+\nabla \Psi)^{2}-V$. 
	
	Notice that $curl \nabla \Psi=0$, so the magnetic field is unchanged by the change of gauge.
	
	Let 
	\begin{equation}\label{eq:sae2}
	W(x_{1})=\|V(x_{1},.)\|_{L^{\infty }(0,d)}=\sup_{0<x_{2}\leq d}|V(x_{1})|
	\end{equation} 
	so that
	\begin{equation}\label{eq:sae3}
	\|W\|_{L^{1}(\mathbb{R},L^{\infty}(0,d))}=\|V\|_{L^{1}(\mathbb{R},L^{\infty}(0,d))}<\infty.
	\end{equation}
	The operator $H_{B,W}=(\frac{1}{i}\nabla+A)^{2}-W$, is lower-semi bounded self adjoint operator with essential spectrum $[0,\infty)$ \cite{balinsky2001number}.	Since $H_{B,V}\geq H_{B,W}$, then
	\begin{equation}\label{eq:sae3a1}
	Neg(H_{B,V})\leq Neg(H_{B,W}),
	\end{equation} 
	where $Neg(.)$ is the number of negative eigenvalues of the operator. Therefore to estimate the number of negative eigenvalues of the operator $H_{B,V}$, its enough to estimate the number of negative eigenvalues of the operator $H_{B,W}$.
	
	Now,
	\begin{equation*}
	\begin{aligned}
	H_{B,W}&=(\frac{1}{i}\nabla+A)^{2}-W\\
	&=-\Delta +\frac{2A}{i}.\nabla+|A|^{2}-W\\
	&=-\frac{\partial^{2}}{\partial x_{1}^{2}}-\frac{\partial^{2}}{\partial x_{2}^{2}}+\frac{2A}{i}.\bigg(\frac{\partial}{\partial x_{1}}+\frac{\partial}{\partial x_{2}}\bigg)+|A|^{2}-W\\
	&=-\frac{\partial^{2}}{\partial x_{1}^{2}}+\frac{2A}{i}\frac{\partial}{\partial x_{1}}-\frac{\partial^{2}}{\partial x_{2}^{2}}+\frac{2A}{i}\frac{\partial}{\partial x_{2}}+|A|^{2}-W.
	\end{aligned}
	\end{equation*}
	Thus 
	\begin{equation}\label{eq:sae4}
	H_{B,W}=-\frac{\partial^{2}}{\partial x_{1}^{2}}+\frac{2A}{i}\frac{\partial}{\partial x_{1}}+\bigg(-\frac{\partial^{2}}{\partial x_{2}^{2}}+\frac{2A}{i}\frac{\partial}{\partial x_{2}}+|A|^{2}\bigg)-W.
	\end{equation}
	Since $-\frac{\partial^{2}}{\partial x_{2}^{2}}+\frac{2A}{i}\frac{\partial}{\partial x_{2}}+|A|^{2}=\bigg(\frac{1}{i}\frac{\partial }{\partial x_{2}}+A\bigg)^{2}$, \eqref{eq:sae4} can now be written as,
	\begin{equation}\label{eq:sae5}
	H_{B,W}=-\frac{\partial^{2}}{\partial x_{1}^{2}}+\frac{2A}{i}\frac{\partial}{\partial x_{1}}+\bigg(\frac{1}{i}\frac{\partial }{\partial x_{2}}+A \bigg)^{2}-W.
	\end{equation}
	Define the operator 
	\begin{equation}\label{eq:sae6}
	\mathcal{Q}:=\frac{1}{i}\frac{\partial }{\partial x_{2}}+A.
	\end{equation}
	Let $\lambda $ be an eigenvalue of $\mathcal{Q}$ and $f$ be the corresponding eigenfunction . Then
	\begin{equation}\label{eq:sae7}
	\mathcal{Q}f=\frac{1}{i}\frac{\partial f }{\partial x_{2}}+Af=\lambda f,
	\end{equation}
	So that
	\begin{equation*}
	\frac{1}{i}\frac{df}{dx_{2}}=(\lambda -A)f.
	\end{equation*} 
	On separating variables and integrating, we obtain
	\begin{equation*}
	\int \frac{1}{f}df=\int i(\lambda-A)dx_{2},
	\end{equation*}
	implying that
	\begin{equation}\label{eq:sae8}
	\ln f=i(\lambda-A)x_{2}+c_{1},
	\end{equation}
	where $c_{1}$, is a constant of integration, 
	\begin{equation}\label{eq:sae9}
	f(x_{2})=ce^{i(\lambda-A)x_{2}}, c=\ln c_{1}
	\end{equation}
	To determine $c$, $\lambda$ and $f$, we use the following boundary conditions.
	\begin{equation}\label{eq:sae10}
	-f^{\prime}(0)=f^{\prime}(d)=0.
	\end{equation} 
	Differentiating \eqref{eq:sae9} with respect to $x_{2}$ gives
	\begin{equation}\label{eq:sae11}
	f^{\prime}(x_{2})=i(\lambda-A)ce^{i(\lambda-A)x_{2}}.
	\end{equation}
	Now putting the boundary conditions \eqref{eq:sae10} into \eqref{eq:sae11}, we obtain
	\begin{equation*}
	i(A-\lambda)c=0
	\end{equation*}
	and
	\begin{equation*}
	i(\lambda-A)ce^{i(\lambda-A)d}=0.
	\end{equation*}
	It follows that
	\begin{equation}\label{eq:sae12}
	\lambda=A \text{ or } e^{i(\lambda-A)d}=0.
	\end{equation}
	Since $e^{i(\lambda-A)d}\ne 0$, then $\lambda=A$. So from \eqref{eq:sae9}, we have $f(x_{2})=c$. Normalizing $f(x_{2})=c$ gives $c=\frac{1}{\sqrt{d}}$.
	
	Define $\varphi$ by
	\begin{equation}\label{eq:varphi}
	\varphi=\min_{k\in{\mathbb{Z}}}|\Psi -k|
	\end{equation}
	Without loss of generality,  assume that $A=\varphi$, where $\varphi$ is defined in \eqref{eq:varphi}. Then we shall obtain that 
	\begin{equation}\label{eq:sae15}
	f_{k}(x_{2})=\frac{1}{\sqrt{d}}e^{i(\lambda_{k}-\varphi)x_{2}}.
	\end{equation}
	Using the decomposition
	\begin{equation*}
	L^{2}(S)=L^{2}(\mathbb{R}, dx_{1})\otimes L^{2}(0,d),
	\end{equation*}
	we have
	\begin{equation}\label{eq:sae16}
	L^{2}(S)=\bigoplus_{k\in{\mathbb{Z}}}\bigg(L^{2}(\mathbb{R}, dx_{1})\otimes\bigg[\frac{e^{i(\lambda_{k}-\varphi)x_{2}}}{\sqrt{d}} \bigg]\bigg),
	\end{equation}
	where $[.]$ denotes the linear span, since $\bigg\{e^{i(\lambda_{k}-\varphi)x_{2}}\bigg\}_{k\in{\mathbb{Z}}}$ is a complete orthonormal system of eigenfunctions in $L^{2}(0,d)$. It follows from \eqref{eq:sae5} that     
	\begin{equation}\label{eq:sae17}
	H_{B,W}=-\frac{\partial^{2}}{\partial x_{1}^{2}}+\frac{2\varphi}{i}\frac{\partial}{\partial x_{1}}+(k+\varphi)^{2}-W=\bigoplus_{k\in{\mathbb{Z}}}\bigg(H_{k}\otimes 1_{k}\bigg),
	\end{equation}
	where 
	\begin{equation}\label{eq:sae18}
	H_{k}=-\frac{\partial^{2}}{\partial x_{1}^{2}}+\frac{2\varphi}{i}\frac{\partial}{\partial x_{1}}+(k+\varphi)^{2}-W.
	\end{equation}
	and $1_{k}$ is the identity on \eqref{eq:sae15}.
	
	For any $u\in{L^{2}(0,d)}$, $u$ has the representation 
	\begin{equation*}
	u(x_{1},x_{2})=\sum_{k\in{\mathbb{Z}}}u_{k}(x_{1})f_{k}(x_{2}),
	\end{equation*}
	where 
	\begin{equation*}
	u_{k}(x_{1})=\int_{0}^{d}u(x_{1},x_{2})f_{k}(x_{2})dx_{2}.
	\end{equation*}
	Thus for $u\in{H_{A}^{1}}$, the quadratic form of \eqref{eq:sae18} given by
	\begin{equation*}
	\begin{aligned}
	q_{k,W}[u]:=\sum_{k\in{\mathbb{Z}}}&\langle -u_{k}^{\prime \prime}(x_{1})f_{k}(x_{2})+\frac{2\varphi}{i}u_{k}^{\prime}(x_{1})f_{k}(x_{2})+(k+\varphi)^{2}u_{k}(x_{1})f_{k}(x_{2})\\
	&-Wu_{k}(x_{1})f_{k}(x_{2}),u_{k}(x_{1})f_{k}(x_{2})\rangle_{L^{2}(0,d)}.
	\end{aligned}
	\end{equation*}
	Thus
	\begin{equation*}
	\begin{aligned}
	q_{k,W}[u]&=\sum_{k\in{\mathbb{Z}}}\bigg(-\langle u_{k}^{\prime \prime}(x_{1})f_{k}(x_{2}),u_{k}(x_{1})f_{k}(x_{2})\rangle +\frac{2\varphi}{i}\langle u_{k}^{\prime}(x_{1})f_{k}(x_{2}),u_{k}(x_{1})f_{k}(x_{2})\rangle \\
	&+(k+\varphi)^{2}\langle u_{k}(x_{1})f_{k}(x_{2}),u_{k}(x_{1})f_{k}(x_{2})\rangle-W\langle u_{k}(x_{1})f_{k}(x_{2}),u_{k}(x_{1})f_{k}(x_{2})\rangle \bigg)\\
	&=\sum_{k\in{\mathbb{Z}}}\bigg(-\int_{S}u_{k}^{\prime \prime}(x_{1})f_{k}(x_{2})\overline{u_{k}(x_{1})f_{k}(x_{2})}dx_{1}dx_{2}+\frac{2\varphi}{i}\int_{S}u_{k}^{\prime}(x_{1})f_{k}(x_{2})\overline{u_{k}(x_{1})f_{k}(x_{2})}dx_{1}dx_{2}\\
	&+(k+\varphi)^{2}\int_{S}u_{k}(x_{1})f_{k}(x_{2})\overline{u_{k}(x_{1})f_{k}(x_{2})}dx_{1}dx_{2}-W\int_{S}u_{k}(x_{1})f_{k}(x_{2})\overline{u_{k}(x_{1})f_{k}(x_{2})}dx_{1}dx_{2} \bigg).
	\end{aligned}
	\end{equation*} 
	It therefore follows that
	\begin{equation}\label{eq:sae18a}
	q_{k,W}[u]=\sum_{k\in{\mathbb{Z}}}\bigg( \int_{\mathbb{R}}\bigg(|u_{k}^{\prime}(x_{1})|^{2}+|k+\varphi|
	^{2}|u_{k}(x_{1})|^{2}-W|u_{k}(x_{1})|^{2}\bigg)dx_{1}\bigg)
	\end{equation}	
\end{section}
\newpage
\begin{section}{The main result}\label{sect4.2} 
	\begin{theorem}[ cf.Theorem \ref{thembalinsky2001}]\label{themneg1}
		Let  $X=L^{1}(\mathbb{R},L^{\infty}(0,d))$, $A$ be the vector potential in \eqref{eq:proofneg2}, $V$ be integrable on bounded subsets of $S$ and $V\in{X}$. Then $Neg(H_{B,V})$ is bounded above by 
		\begin{equation}\label{eq:result1clr4}
		Neg(H_{B,V})\leq  {\sum}^{\prime}_{k\in{\mathbb{Z}}}\frac{1}{\sqrt{16|k+\varphi|^{2}+1}}\|V\|_{X},
		\end{equation}
		where ${\sum}^{\prime}$ indicates that all summands less than $1$ are neglected. 
	\end{theorem}
   
   \begin{proof}
   	The operator $H_{B,W}$ is defined via its quadratic form \eqref{eq:sae18a}, that is, 
   	\begin{equation}\label{eq:proofneg8}
   	q_{k,W}[u]=\sum_{k\in{\mathbb{Z}}}\bigg(\int_{\mathbb{R}}\bigg( |u^{\prime}_{k}(x_{1})|^{2}+(k+\varphi)^{2}|u_{k}(x_{1})|^{2}-W(x_{1})|u_{k}(x_{1})|^{2}\bigg)dx_{1}\bigg).
   	\end{equation} 
   	By the Hardy inequality \eqref{eq:onedimhardy1} (also \cite[ Theorem 327]{hardy1952inequalities}), we have
   	\begin{equation*}
   	q_{k,W}[u]\geq \sum_{k\in{\mathbb{Z}}}\int_{\mathbb{R}}\bigg(\frac{1}{4} |u^{\prime}(x_{1})|^{2}+ \frac{(k+\varphi)^{2}}{x_{1}^{2}}|u(x_{1})|^{2}-W(x_{1})|u_{k}(x_{1})|^{2}\bigg)dx_{1} 
   	\end{equation*}
   	\begin{equation}\label{eq:proofneg101}
   	=\frac{1}{4} \sum_{k\in{\mathbb{Z}}}\int_{\mathbb{R}}\bigg(|u^{\prime}(x_{1})|^{2}+ 4\frac{(k+\varphi)^{2}}{x_{1}^{2}}|u(x_{1})|^{2}-4W(x_{1})|u(x_{1})|^{2}\bigg)dx_{1}.
   	\end{equation} 
   	The right hand of \eqref{eq:proofneg101} is a quadratic form associated with a family of Sturm-Liouville operators of the form
   	\begin{equation}\label{eq:proofneg105}
   	\tau(k):=\frac{1}{4}\bigg(-\frac{d^{2}}{dx_{1}^{2}}+\frac{4(k+\varphi)^{2}}{x_{1}^{2}}-4W(x_{1})\bigg) \text{  on $L^{2}(\mathbb{R})$}.
   	\end{equation}
   	Now, it follows from \eqref{eq:bargmann3} that 
   	\begin{equation}\label{eq:proofneg105b}
   	Neg(\tau(k))\leq\frac{1}{\sqrt{16(k+\varphi)^{2}+1}}\int_{\mathbb{R}}W(x_{1})dx_{1}	
   	\end{equation}   	
   	Hence, by \eqref{eq:sae3a} and  \eqref{eq:proofneg101}, we have
   	\begin{equation}\label{eq:proofneg105c}
   	Neg(H_{B,W})\leq{\sum}^{\prime}_{k\in{\mathbb{Z}}}\frac{1}{\sqrt{16|k+\varphi|^{2}+1}}\int_{\mathbb{R}}W(x_{1})dx_{1}.	
   	\end{equation}
   	Substituting for $W$ (see \eqref{eq:sae3a1}) in \eqref{eq:proofneg105c} yields
   	\begin{equation*}\label{eq:proofneg15}
   	\begin{aligned}
   	Neg(H_{B,V})&\leq{\sum}^{\prime}_{k\in{\mathbb{Z}}}\frac{1}{\sqrt{16(k+\varphi)^{2}+1}}\int_{\mathbb{R}}ess\sup_{0<x_{2}\leq d}V(x_{1},x_{2})dx_{1}\\
   	&={\sum}^{\prime}_{k\in{\mathbb{Z}}}\frac{1}{\sqrt{16|k+\varphi|^{2}+1}}\|V\|_{X}.
   	\end{aligned}
   	\end{equation*}
   \end{proof}	
\end{section}
\begin{section}{Dependence of the estimate on the magnetic field.}\label{sect4.4}
	The estimate \eqref{eq:result1clr4} will only hold if the magnetic field is present. In the absence of the magnetic field the operator $H_{BV}$ reduces to the usual electric two dimensional Schr\"{o}dinger operator $H_{V}=-\Delta-V$, $V\geq 0$ on $L^{2}(S)$.	One of the reasons why \eqref{eq:result1clr4} fails is that the Hardy inequality does not hold even for radially symmetric functions.
	
	For example, take $u\in{C_{0}^{\infty}(0,\infty)}$ such that $u(r)=\ln{\ln{\frac{1}{r}}}$ for $0<r<\frac{1}{e}$. then the inequality 
	\begin{equation}\label{eq:dis1}
	\int_{0}^{\infty}|u(r)|^{2}\frac{dr}{r}\leq c\int_{0}^{\infty}|u^{\prime}(r)|^{2}dr.
	\end{equation}
	clearly fails. The implication of this is that \eqref{eq:proofneg101} in the proof of theorem \eqref{themneg1} fails, and consequently \eqref{eq:result1clr4} does not hold.
	
	Therefore theorem \eqref{themneg1} holds true only in presence of the magnetic field which does not concentrate at integer points.	
\end{section}	
\end{chapter}
\begin{chapter}{Conclusion and Recommendations}
\begin{section}{Conclusion}
	In this study, the spectral estimate of the magnetic Schr\"{o}dinger operator with Aharonov-Bohm magnetic field in a straight strip subject to Neumann boundary conditions has been obtained. In particular, we have obtained an upper estimate \eqref{eq:result1clr4} for the number of negative eigenvalues for the operator \eqref{eq:result1clr1}. We have shown in \S \eqref{sect4.4} that the upper estimate \eqref{eq:result1clr4} depends on the properties of the magnetic field $B$ and that such an estimate fails in the absence of magnetic field and only holds in the presence of the Aharonov-Bohm magnetic field with the flux satisfying the non integer condition.
\end{section}
\begin{section}{Recommendations}
	There are several directions in which one can extend the above results, for instance, treating the problem using Numerical techniques, one can investigate whether mixed (Dirichlet, Neumann and Robin) boundary conditions applied on the strip can generate similar results.
	
	An open problem here concerns the optimal constant (what is the best constant?) in the estimate. In an attempt to answer this question, one has two options; one is either treating it analytically or using Numerical techniques.
	
	Also, one can extend the above results to a curved quantum waveguide with magnetic fields, where in this case the curvature also contributes to the electric potential.
	
	Finally, instead of restricting the magnetic field to non-integer flux values, it would be interesting to consider a reasonably large class of fields by removing the singularities at integer points.  
\end{section}

\end{chapter}

\end{document}